\documentclass{article}
\usepackage[utf8]{inputenc}
\usepackage[margin=1in]{geometry}
\usepackage{cancel}
\usepackage{caption}
\usepackage{subcaption}
\usepackage{amsthm}
\usepackage{amsmath}
\usepackage{textcomp}
\usepackage{hyperref}
\hypersetup{colorlinks,allcolors=blue}
\usepackage{amsfonts}
\usepackage{mathtools}
\usepackage{amssymb}

\usepackage{fancyhdr}
\usepackage{enumerate}
\usepackage{bm}
\usepackage{thmtools,thm-restate}
\usepackage[pdftex]{graphicx}
\usepackage{comment}
\usepackage{xspace}
\usepackage{xcolor}
\usepackage[numbers]{natbib}

\usepackage{algorithm}
\usepackage[noend]{algpseudocode}

\newtheorem{theorem}{Theorem}%[section]

\newtheorem{lemma}[theorem]{Lemma}
\newtheorem{definition}{Definition}
\newtheorem{remark}[theorem]{Remark}
\newtheorem{fact}{Fact}%[section]
\newtheorem{assumption}{Assumption}
\newcommand{\eps}{\varepsilon}

\newcommand{\MCS}{\textsc{MaxCut}}

\newcommand{\MAS}{\textsc{Mas}}

\newcommand{\BTW}{\textsc{Btw}}
\newcommand{\UGC}{\textsc{Ugc}}
\newcommand{\ML}{\textsc{Must-Link}}
\newcommand{\CL}{\textsc{Cannot-Link}}
\newcommand{\NSEP}{\textsc{4-Non-Separatedness}}
\newcommand{\SEP}{\textsc{4-Separatedness}}
\newcommand{\nBTW}{\textsc{non-Btw}}

\newcommand{\bS}{\bar{S}}
\newcommand{\EE}{\mathbb{E}}
\newcommand{\QF}{\mathcal{Q_F}}
\newcommand{\QD}{\mathcal{Q_D}}
\newcommand{\TF}{\mathcal{T_F}}
\newcommand{\TD}{\mathcal{T_D}}
\newcommand{\QA}{\mathcal{Q_A}}
\newcommand{\TA}{\mathcal{T_A}}
\newcommand{\MC}{\textsc{MaxCut}}
\newcommand{\FAS}{\textsc{Fas}}
\newcommand{\ALG}{\texttt{ALG}}
\newcommand{\SDP}{\texttt{SDP}}
\newcommand{\OPT}{\texttt{OPT}}
\newcommand{\sgn}{\texttt{sgn}}

\title{Maximizing Agreements for Ranking, Clustering\\ and Hierarchical Clustering via MAX-CUT}

\author{
  Vaggos Chatziafratis\thanks{\texttt{vaggos@cs.stanford.edu} and \texttt{vaggos@google.com}}\\
  Google Research NY
  \and
  Mohammad Mahdian\thanks{\texttt{mahdian@google.com}}\\
  Google Research NY
  \and
  Sara Ahmadian\thanks{\texttt{sahmadian@google.com}}\\
  Google Research NY
}

\usepackage{natbib}
\usepackage{graphicx}

\begin{document}

\maketitle

\begin{abstract}
    In this paper, we study a number of well-known combinatorial optimization problems that fit in the following paradigm: the input is a collection of (potentially inconsistent) local relationships between the elements of a ground set (e.g., pairwise comparisons, similar/dissimilar pairs, or ancestry structure of triples of points), and the goal is to aggregate this information into a global structure (e.g., a ranking, a clustering, or a hierarchical clustering) in a way that maximizes agreement with the input. Well-studied problems such as rank aggregation, correlation clustering, and hierarchical clustering with triplet constraints fall in this class of problems. 
We study these problems on stochastic instances with a hidden embedded ground truth solution. Our main algorithmic contribution is a unified technique that uses the maximum cut problem in graphs to approximately solve these problems. Using this technique, we can often get approximation guarantees in the stochastic setting that are better than the known worst case inapproximability bounds for the corresponding problem. On the negative side, we improve the worst case inapproximability bound on several hierarchical clustering formulations through a reduction to related ranking problems.
\end{abstract}
\section{Introduction}

In many learning/optimization problems, the input data is in the form of a number of {\em ordinal} judgements about the local relationships among a set of $n$ items. A prominent example is the problem of ranking $n$ alternatives, where the input is often pairwise comparisons between these items. For example, sports teams are often ranked by aggregating the results of matches played between pairs of teams, and election outcomes are decided by aggregating individual votes. 

Learning from comparisons has been prevalent across different domains, as humans are typically good at quickly answering {\em ordinal} questions (``which movie/restaurant/candidate do you prefer''), but often respond slowly and inaccurately to {\em cardinal} questions (``how much do you like this option''). In the psychology literature, the method of {\em paired comparisons} that has been in use since the 1920's is based on this principle (see~\cite[Chapter 7]{pairedcomp}). Moreover, modern online platforms can organically extract such ordinal preferences by observing the users (e.g., ``which movie did they first watch'', or ``did they skip a search result and click on the next one'') and later use them for improving search or recommendation rankings (see, for example,~\cite{joachims}). The same principle applies to settings other than ranking. For example, when trying to learn a clustering of $n$ items, it is easier for a human judge to answer questions of the form ``should $x$ and $y$ be in the same cluster'' than to measure the similarity of $x$ and $y$. Or, to reconstruct the evolutionary tree (also known as the phylogenetic tree) between $n$ species, biologists often start by answering questions of the form ``between three species $x, y$, and $z$, which two are evolutionarily closer''. 

At the heart of each of these examples is the non-trivial algorithmic task of reconciling potentially inconsistent judgements into a global solution. This defines a number of algorithmic problems that we study in this paper. Though seemingly unrelated, all of these problems seek to find a global structure that has the maximum number of agreements with the given collection of local ordinal relationships. As we shall see later in the paper, the problems are also linked in that we can apply a common technique (based on graph max cut) to them all. The problems, shown in  Figure~\ref{fig:examples}, fall under the three categories of ranking, clustering, and hierarchical clustering:

\begin{itemize}
\item {\bf Ranking:} The goal is to find an ordering of $n$ items. In the {\em Maximum Acyclic Subgraph (\MAS)}, the input is a number of pairwise comparisons of the form $a<b$. In {\em Betweenness}, the input is a number of triples $a|b|c$ meaning that $b$ is between $a$ and $c$ in the ordering. In {\em Non-Betweenness}, the input is a number of triples $b|ac$ meaning that $b$ is not between $a$ and $c$.
\item {\bf Clustering:} In the {\em Correlation Clustering} problem, the goal is to find a partitioning of $n$ items, and the input is a number of pairs of the form $ab$, meaning that $a$ and $b$ should be in the same cluster, and a number of pairs of the form $a|b$, meaning that $a$ and $b$ should be in different clusters.
\item {\bf Hierarchical clustering:} The goal is to find a (rooted or unrooted) tree with the set of $n$ items as its leaves. In the {\em Desired Triplets} problem, the input is a number of triplets $ab|c$, meaning that the least common ancestor of $a$ and $b$ is a descendant of the least common ancestor of $a, b$, and $c$.  In the {\em Desired Quartets} problem, the input is a number of quartets $ab|cd$, meaning that the unique path connecting $a$ and $b$ in the tree does not intersect with the unique path connecting $c$ and $d$. The {\em Forbidden Triplets} and {\em Forbidden Quartets} problems are defined similarly with the opposite requirements.
\end{itemize}

These problems come from a variety of applications: \MAS\ is a formulation of the rank aggregation problem and has many applications, e.g., in search ranking. Correlation Clustering is a central problem in unsupervised learning and data analysis~\citep{bansal2004correlation}. Hierarchical clustering problems are motivated by applications in reconstructing phylogenetic trees~\citep{felsenstein2004inferring}, and are also related to the objective-driven formulations of~\cite{dasguptaHC}, \cite{moseley2017approximation} and~\cite{cohen2019hierarchical} for hierarchical clustering. In fact, the Desired Triplets formulation described above is tightly connected with objective-based approaches for Hierarchical Clustering as can be seen in~\cite{charikar2019hierarchical,charikar2019euclidean}.
Betweenness and Non-Betweenness are motivated by applications in genome sequencing in bioinformatics~\citep{slonim1997building}. We are interested in algorithms that can provide an approximation guarantee, i.e., a provable bound on the multiplicative factor between the solution found by the algorithm and the optimal solution. 
We will consider this problem both in the worst case and under a stochastic model with an embedded ground-truth solution.

\paragraph{Main Results:} Our contribution is two-fold (see Table~\ref{tab:pos} for a summary): On the positive side, in Section~\ref{sec:pos}, under a simple stochastic model akin to the well-known {\em stochastic block model}, we are able to improve upon worst-case approximations for all problems and in some cases (e.g., for problems on rankings and hierarchies) even overcome impossibility results. Interestingly, our algorithms are all based on variants of \MC\ on graphs that can have both positive and negative weights and may also be directed. Some approaches for tree reconstruction based on \MC\ had been used in previous experimental works~\citep{snir2006using,snir2008quartets,snir2012quartet}, and in this way our work provides concrete proof for why these heuristics are reported to perform well on ``real-world'' instances.  Our natural stochastic model captures ``real-world'' instances via an embedded ground-truth from which we generate ``noisy''  constraints, similar to the  Stochastic Block Model~\citep{mossel2012stochastic} in community detection. 

%\item 
On the negative side, we obtain new hardness of approximation results for four problems on hierarchical clustering: Forbidden Triplets, Desired Triplets, Forbidden Quartets, Desired Quartets. Briefly, we may refer to them as triplets/quartets consistency problems. These are instances of Constraint Satisfaction Problems (CSP) on trees~\citep{bodirsky2010complexity,bodirsky2016complexity}, analogous to SAT formulas in complexity. Even though such problems on hierarchies have been studied for decades, the current best approximations are achieved by trivial baseline algorithms. Our hardness results give some explanation why previous approaches  were not able to obtain anything better. Our result on the Forbidden Triplets problem is tight and is the first tight hardness for CSPs on trees, extending analogous hardness results by~\cite{guruswami2011beating} from linear orderings (i.e., rankings) to trees. This is carried out in Section~\ref{sec:neg}.
    
%\end{itemize}

Our stochastic model for collecting information is the simplest form of embedded model on $n$ items, and is motivated by crowdsourcing and biological applications~\citep{vaughan2017making, kleindessner2017kernel,ghoshdastidar2019foundations,snir2012reconstructing}. We simply choose items at random and include a pairwise/triplet/quartet constraint depending on the task. For example, to generate constraints for the \MAS\ problem on rankings, let $\pi^*$ denote a ground-truth ranking (e.g., of chess players or ads to show a user). We select uniformly at random $m$ pairs of items $a_i,b_i$ and then we generate $m$ pairs $a_i<b_i$; if $a_i$ precedes $b_i$ in $\pi^*$ the constraint is included with probability $(1-\eps)$, otherwise the opposite constraint is generated. Thus, some fraction of the constraints can be erroneous. After generating $m$ (noisy) constraints in this way, our goal is to find a global solution (ranking, partition, or tree) that satisfies as many as possible. 
\begin{figure}[ht]
  \centering
  \includegraphics[width=\linewidth]{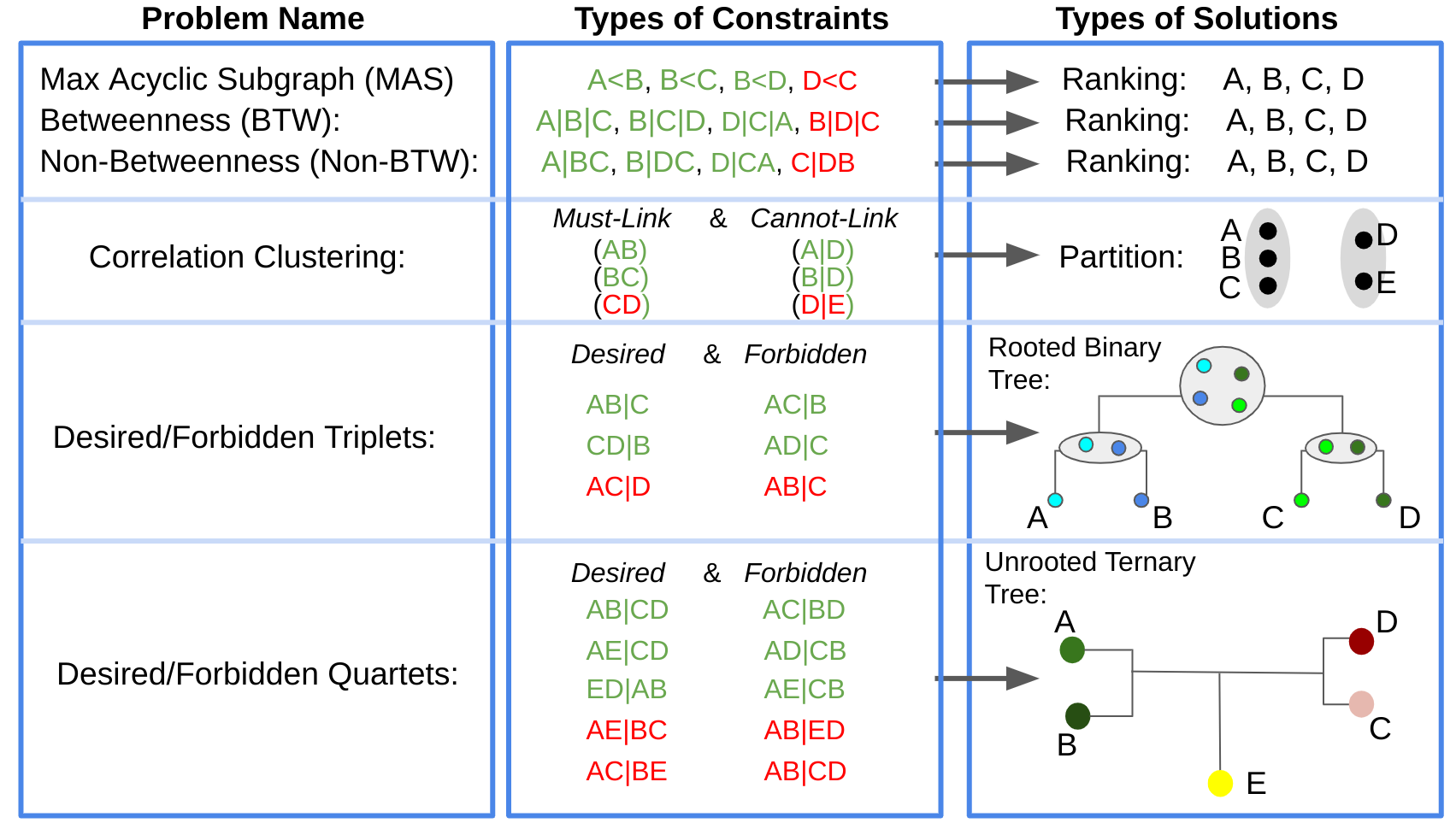}
     \caption{A schematic representation of all problems considered in the paper. The left column has the problem names, the middle the types of constraints and the right column has a candidate solution. With {\color[rgb]{0,0.6,0} green} are constraints that are correctly resolved in the given candidate solution, whereas with {\color{red} red} are those that are incorrect. For more examples, see Section~\ref{sec:problems}.}
     \label{fig:examples}
\end{figure}

\paragraph{Techniques:} Our hardness reductions for Maximum Forbidden Triplets consistency are based on mapping trees to permutations on their leaves and back, and showing that any constant factor improvement over trivial baselines would refute the Unique Games Conjecture\footnote{Khot's \UGC\ is a major open question in complexity. We will not define it here as we only use some of its consequences on ordering problems~\citep{guruswami2011beating}.} (\UGC)~\citep{khot2002power}. Regarding our \MC\ algorithm~(see Algorithm~ \ref{alg:decoding}), it is based on \MC\ variations on directed and undirected graphs with negative weights and is conceptually simple. Briefly, given an instance for any of the problems we consider, we map it to a graph where edges encode the underlying constraints; perhaps the most intuitive such construction is for Correlation Clustering where a ``must-link'' or ``cannot-link'' constraint between items $i,j$ is captured by a negative or positive edge $(i,j)$ respectively. Then, we show how large (positive) cuts in this graph yield partitions that satisfy many of the constraints. The existence of a large cut can be guaranteed by analyzing our stochastic model and so an approximate \MC\ algorithm can yield improvements over previous results.  An interesting ingredient that we need for the case of \MAS, is how to approximate the \MC\ problem on \emph{directed} graphs with both positive and negative weights which, to the best of our knowledge, hadn’t been analyzed before.

More broadly,  we justify theoretically why prior experimental heuristics work and we extend them to work for new problems with provable approximation guarantees. Our work also presents the first case of a CSP on trees that is approximation resistant; recall that many important CSPs, including Max3SAT, are approximation resistant, i.e., it is NP-hard to approximate them better than a random assignment. This echoes the striking result by~\cite{haastad2001some} on approximation resistance of boolean CSPs to CSPs on trees and shows why no algorithmic improvement had been made in the worst-case, despite significant efforts~\citep{byrka2010new,jiang2001polynomial, bryant1997building, he2006inferring, steel1992complexity}. 

%In fact,  our results provide strong evidence leading us to believe the following conjecture:

%\begin{conjecture*}
%For Maximum Triplets Consistency and Maximum Quartet Consistency problems, one cannot beat, in polynomial time, the trivial approximation factor obtained by a random tree.
%\end{conjecture*}

%We resolve the conjecture 

%  Our conjecture 

%Finally, byproducts of our work include a PTAS for \emph{dense} instances of recent maximization formulations of Hierarchical Clustering objectives~\cite{joshNIPS,cohen2019hierarchical} and a $O(\log n\log \log n)$-approximation for Minimum Triplets Inconsistency (\RTI) on caterpillar trees (see Section~\ref{sec:rti}). We mention here that for the general case of the \RTI\ problem, known hardness results~\cite{chester2013resolving} prevent us unfortunately from getting anything better than $n$-approximation achieved by a variant of Minimum Cut~\cite{jansson2003consensus}. 

\begin{table}[!ht]  
  \centering
  \caption{\label{tab:result} Shown in bold are our improved hardness (column ``Hardness'') and approximations under our stochastic model (column ``Stochastic''). Column ``Approx.'' has prior approximation ratios. Also see Section~\ref{sec:pos} and Appendix~\ref{app:beating} for the dependence on error parameter $\eps$.}\label{tab:pos}
  \vspace{0.1cm}
  \begin{tabular}{|c||c|c|c|}
    \cline{2-4}
    \multicolumn{1}{c|}{} & Approx. & Hardness & Stochastic \\ \hline 
\MAS & 1/2 & 1/2 & {\bf0.642}  \\ \hline
\BTW & 1/3 & 1/3 & {\bf 0.402} \\ \hline
\nBTW & 2/3 & 2/3 & {\bf 0.84} \\ \hline
Correl. Cl. & 0.76 & APX-hard & {\bf0.82}(*) \\ \hline
Forb. Triplet & 2/3 & \textbf{2/3 (tight)} & {\bf 0.78}(*) \\
\hline
Des. Triplet & 1/3 & \textbf{2/3} & {\bf0.64}(*)  \\ \hline
Forb. Quartet & 2/3 & \textbf{8/9 }& {\bf0.672} \\ \hline

Des. Quartet & 1/3 & \textbf{2/3} & 0.425 \\ \hline

  \end{tabular}
\end{table}

%\begin{table}[!ht] 
%  \centering
%  \caption{\label{tab:result}Improved worst-case hardness, and approximations under our sampling model for hierarchical clustering problems, shown in bold (for exact dependency on $\eps$, see Section~\ref{sec:pos} and Appendix~\ref{app:beating}).} \label{tab:neg}
%  \vspace{0.1cm}
%  \begin{tabular}{|c||c|c|c|}
%    \cline{2-4}
%    \multicolumn{1}{c|}{} & Upper & Lower & \textsc{Max-Cut} \\ \hline
%Desired Quartets & 1/3 & \textbf{2/3} & 0.425 \\ \hline
%Forbidden Quartets & 2/3 & \textbf{8/9 }& {\bf0.672} \\ \hline
%Desired Triplets & 1/3 & \textbf{2/3} & {\bf0.64}(*)  \\ \hline
%Forbidden Triplets & 2/3 & \textbf{2/3} & {\bf 0.78}(*) \\ \hline

%  \end{tabular}
  
%\end{table}

\begin{remark}
We want to point out that all our approximation results here hold with high probability as a standard concentration argument about the stochastic process guarantees that the weight of the cuts is well-concentrated around its mean (as long as the number of generated constraints $m\ge \Omega(\log n)$).
\end{remark}

\begin{remark} 
Our results for ranking and quartets hold with no assumption on the optimal solution. For the positive results (denoted with (*) in Table~\ref{tab:pos}) via \MCS\ for correlation clustering and triplets however, we need a mild balancedness assumption, roughly stating that the optimal solution contains a relatively balanced $\left(\tfrac13:\tfrac23\right)$ partition, to ensure the existence of a good cut in the ground-truth (see Appendix, Assumption~\ref{ass:triplets}). Usually, such assumptions are common in generative graph models for clustering, e.g., the Stochastic Block Model~\citep{mossel2012stochastic,abbe2015exact} and for hierarchical clustering, e.g., the Hierarchical Stochastic Block Model~\citep{lyzinski2016community,cohen2019hierarchical,ghoshdastidar2019foundations}, where we expect to see at least two large communities emerge.
\end{remark}

\section{Background and Related Work}\label{sec:problems}

As the paper discusses multiple problems on rankings, partitions and hierarchies, we devote this section in describing the multitude of problems. A familiar reader can skip this section and proceed to Section~\ref{sec:pos}.

There are 3 categories of problems we study here, depending on the type of the output:  ranking (also called a \emph{permutation} or a \emph{leaf ordering} in biology~\citep{bar2001fast}), clustering (partitioning of the data points) and hierarchical clustering (also called \textit{phylogenetic tree}). There has been significant amounts of work on each of these tasks, that we only partially cover here as we go over our problems and results.

\subsection{Optimization Problems and Types of Constraints}
 In all problems, we are given $m$ constraints and we want to maximize the number of constraints satisfied by our output, whether it be a ranking, a partition or a hierarchy. We describe below the types of different constraints (see also Figure~\ref{fig:examples}):

\paragraph{Ranking (i.e., a permutation or leaf ordering):}  Given $n$ labels $\{1,2,\ldots, n\}$, we want to find a permutation that maximizes the number of satisfied constraints  of the following form:
    \begin{itemize}

    \item \textbf{Pairwise comparisons}: A constraint here is of the form ``$a<b$'', indicating that in the output permutation, item $a$ should precede $b$. If this information is encoded as a directed graph $G$ with arcs $a\to b$, this gives rise to the Maximum Acyclic Subgraph (\MAS) or Feedback Arc Set (\FAS), two fundamental problems in computer science~\citep{karp1972reducibility}.
    
%\vspace{-2mm}
    \item \textbf{Betweenness (BTW) and Non-Betweenness (Non-BTW) constraints}: In the \BTW\ problem~\citep{opatrny1979total, chor1998geometric,makarychev2012simple}, we are given relative ordering constraints of the form $a|b|c$ indicating ``$b$ should be between $a$ and $c$''. This allows for $abc$ or $cba$ out of the 6 possible orderings for the 3 labels. As the name suggests, \nBTW\ is the complement of \BTW, where a constraint $bc|a$ (equivalently $a|bc$) indicates that in the output permutation ``$a$ should \emph{not} lie between $b$ and $c$''. This allows for 4 valid relative orderings $abc,acb,bca,cba$. Generally, these are the two most common examples of ordering Constraint Satisfaction Problems (ordering CSPs) of arity 3 and are mainly motivated by applications in bioinformatics~\citep{slonim1997building}. They have also played a major role in complexity~\citep{guruswami2011beating, austrin2013np}. 
\end{itemize}
%\vspace{-2mm}

Just to give a sense of the approximability of these problems in the worst-case, the current best constant factor is a $\tfrac12$-approximation for \MAS, a $\tfrac13$-approximation for \BTW, and a $\tfrac23$-approximation for \nBTW, all achieved by a \emph{random} permutation. We also know that under the Unique Games Conjecture (\UGC) of~\cite{khot2002power}, the first two results are tight, whereas the third is tight under P $\neq$ NP. Such problems, where a random output is provably the best, are called \textit{approximation resistant} and have been studied extensively by theoreticians~\citep{charikar2009every,guruswami2008beating,haastad2001some,austrin2009approximation}. Our work gives strong evidence pointing to the fact that important CSPs on trees (triplets/quartets) may be approximation resistant.

\paragraph{Clustering:} Here we want to maximize agreements with \textbf{Must-Link/Cannot-Link} constraints: The input is a graph with ``$+$'' or ``$-$'' edges indicating if the two endpoints should belong to the same cluster or not. Such constraints give rise to Correlation Clustering, an important paradigm for data analysis both in practice~\citep{davidson2007survey,wagstaff2000clustering,wagstaff2001constrained} and theory~\citep{bansal2004correlation,ailon2008aggregating,charikar2005clustering,swamy2004correlation}. The current best for maximizing agreements is a $0.7666$   multiplicative approximation via semidefinite programs~\citep{swamy2004correlation} and an APX-hardness is known~\citep{charikar2005clustering}. Here we will improve upon $0.7666$, under our stochastic model for generating constraints.
   %mention it captures MAX CUT, coloring etc.

%\footnote{(APX-hard means it is NP-hard to approximate within some constant factor strictly smaller than 1.}
\paragraph{Hierarchical Clustering (i.e., phylogenetic trees):} There are two common types of trees: rooted and unrooted. Given $n$ data points, a rooted binary tree on $n$ leaves, where each leaf corresponds to a data point, is usually called a \textit{hierarchical clustering} and is a standard tool for data analysis across different disciplines~\citep{steinbach2000,leskovec2014mining,tumminello2010correlation,sorlie2001gene}. Unrooted ternary trees (all nodes have degree 3, except the leaves that have degree 1) are usually called  \textit{phylogenetic trees} and are prevalent in computational biology as they describe speciation events throughout the evolution of species~\citep{bryant1997building,felsenstein2004inferring}. Here we will use the two terms interchangeably to describe hierarchies on $n$ leaves.
%To draw a distinction between the two types of trees, in this paper we will follow the convention and refer to binary rooted trees as hierarchical clusterings, whereas to ternary unrooted trees as phylogenies or phylogenetic trees. 
Since in a hierarchy all data are eventually separated at the leaves, pairwise constraints no longer make sense and the analogue of ``must-link/cannot-link'' are so-called ``must-link-before/cannot-link-before'' constraints:

%\vspace{-2mm}

\begin{itemize}
\item \textbf{Desired/Forbidden Triplets:} The output here is a rooted binary tree $T$ on $n$ leaves. We say a triplet relation ``$t=ab|c$'' is \emph{obeyed} by $T$ (or $T$ \textit{obeys} $t$), if the lowest common ancestor (LCA) of $a,b$ is a descendant of the LCA of $a,c$ in $T$. Otherwise $T$ \textit{disobeys} $ab|c$. A triplet can be \textit{desired} (we write $t\in\TD$) and we want the output $T$ to obey it\footnote{For example, ``penguin, dolphin$|$~tiger'' could be a desired triplet as the tiger is the least relevant item.} or \textit{forbidden} (we write $t\in \TF$) and we want $T$ to disobey/avoid it, giving rise to important optimization problems studied in computational biology and graph theory under the name of rooted triplets consistency~\citep{steel1992complexity,bryant1997building,byrka2010new,he2006inferring}.  Notice that a forbidden triplet $ab|c$ is less restrictive, since it only specifies that $T$ should either obey $ac|b$ or $bc|a$, but not $ab|c$. This is reflected in the complexity of the problems: given a set of forbidden triplets, it is NP-complete to check consistency (i.e., if there is a tree avoiding all of them), whereas checking consistency of desired triplets in polynomial time was established long ago by~\cite{ahoBUILD}. 

%\vspace{-2mm}

\item  \textbf{Desired/Forbidden Quartets:} The desired output here is a ternary unrooted tree $T$. We say a quartet $q=ab|cd$ is \textit{obeyed} by $T$ (or $T$ \textit{obeys} $q$) if the (unique) path from $a$ to $b$ in $T$ does not share any vertices with the (unique) path from $c$ to $d$ in $T$. Otherwise $T$ \textit{disobeys} $q$. Similarly to triplets, a quartet can be \textit{desired} ($q\in \QD$) or \emph{forbidden} ($q\in \QF$), giving rise to important quartets consistency problems in biology and graph theory~\citep{felsenstein2004inferring,bryant1997building,jiang2001polynomial,snir2006using}. For both problems, even if the input is consistent, checking consistency is NP-complete.
\end{itemize}

Once again, just to give a sense of the approximability, for desired triplets or quartets, the current best is a $\tfrac13$-approximation and for forbidden triplets or quartets, the current best is a $\tfrac23$-approximation. Embarrassingly, in all four cases these are achieved by a random (rooted or unrooted) tree or a simple greedy construction~\citep{he2006inferring}.

\subsection{Further Motivation and Related Work}
\label{sec:related}
Here, we further make a comparison to other relevant works. For ranking, many different types of probabilistic models have been considered~\citep{braverman2009sorting,shah2016stochastically,shah2017simple,negahban2012iterative, falahatgar2017maximum} giving statistical guarantees for reconstructing the desired permutation. Instead of pairwise comparisons, the problem has also been studied in the case where partial rankings or complete information (``tournaments'') is provided ~\citep{fagin2006comparing,ailon2010aggregation,kenyon2007rank}. Clustering with constraints and qualitative information (both max and min versions) were studied in~\cite{bansal2004correlation,charikar2005clustering} where approximations via linear programs were derived or practical improvements were made possible~\citep{wagstaff2001constrained,wagstaff2000clustering}. In crowdsourcing and biological applications, both triplet and quartets queries have been deployed~\citep{vinayak2016crowdsourced,vaughan2017making, kleindessner2017kernel,ghoshdastidar2019foundations,snir2006using,bryant1997building} as they can be more intuitive for non-expert users compared to pairwise comparisons. Semi-supervised models, where triplet queries depend on answers to previous queries have been studied in~\cite{kempeSODA, vikram2016interactive}.%, and is known that $O(n\log n)$ adaptive queries suffice to rebuild the optimum tree exactly. 

To further motivate our stochastic model and results, we include a slightly more detailed comparison with 3 important prior works \cite{braverman2009sorting,kempeSODA,snir2012reconstructing} that study ``ground-truth'' stochastic models similar to ours. The authors in~\cite{braverman2009sorting} study the ranking problem and assume that there exists a ground-truth ranking $\pi^*$, as we do. However, their stochastic model assumes either that we have access to \textit{all}  pairwise comparisons, or that we have access to \textit{complete} rankings $\sigma$ on the $n$ items, where each complete ranking $\sigma$ is generated with probability inverse exponential in the Kemeny distance between $\pi^*$ and $\sigma$ (Kemeny distance is the number of inversions, i.e., the number of pairs ordered in $\pi^*$ differently from $\sigma$). 
 
 As it will become obvious, their assumptions are much stricter than our simple stochastic model that generates $m$ pairwise comparisons uniformly at random. Moreover, notice that our approximation guarantees hold for \textit{any} number $m$ of given constraints without requiring it to be $\Omega(n^2)$. Given their more refined model, they are of course in a position to analyze the maximum likelihood estimator and prove approximate recovery results, e.g., that no element is misplaced by more than $\log n$ positions with high probability; however no guarantees are given for the number of violated pairwise constraints, which is the focus of our paper.
 
 For triplets hierarchical clustering, the authors in~\cite{kempeSODA} assume there exists a ground-truth binary tree $T$, as we do. However, they are allowed \textit{adaptive} triplet queries and show that $\approx n\log n$ such queries suffice to recover $T$ using a clever partition algorithm similar to Quickselect and Quicksort. Once again, our model is not adaptive, and we do not pose any constraints on the number $m$ of given constraints. For quartets hierarchical clustering, our model is similar to~\cite{snir2012reconstructing}, but we generalize their results to hold both for forbidden and desired quartets. 
 
 Finally, our constrained version of Hierarchical Clustering based on triplet constraints was studied in~\cite{chatziafratis2018hierarchical} under the assumption that the input contains pairwise similarities as well as triplet constraints.

%\vspace{-2mm}
\section{Using \MCS\ on instances with embedded ground-truth}
%\vspace{-2mm}
\label{sec:pos}
We present our main strategy \MCS\ behind our positive results. As we will see, by modifying the graphs, our method is flexible to allow for combinations of constraints, e.g., both \BTW\ and \nBTW\ constraints for rankings, or both desired and forbidden triplets (or quartets) for trees.

%and also the proof that Forbidden Triplets Consistency is approximation resistant. Then, in later sections we see how to adapt the Max Cut strategy for other problems, and how to get hardness for quartets consistency as well. For the latter, one challenge is how to get a permutation on the labels of an unrooted ternary tree, satisfying constraints on $4$ items that the tree obeys. 

\paragraph{Stochastic Model for Generating Constraints:} Since our goal is to beat the worst-case approximation and hardness results, we use a simple stochastic model with an embedded ground-truth solution on $n$ items. The form of the ground-truth changes depending on which problem we consider; it can be a ranking (for \MAS, \BTW, \nBTW), a partition (for Correlation Clustering) or a hierarchical tree (rooted for Triplets and unrooted for Quartets). For generating the $m$ input constraints, we simply choose items at random and with probability $(1-\eps)$ we add a pairwise/triplet/quartet constraint that is consistent with the ground-truth, otherwise with probability $\eps$ we add an erroneous constraint on the selected items. For example, in the introduction, we saw the \MAS\ constraints. Similarly, for \BTW, we would uniformly at random pick $m$ triples of items $a,b,c$ and then add w.p. $(1-\eps)$ the constraint $a|b|c$ if $b$ appears in between $a$ and $c$ in the ground-truth ordering. Also, for the Triplets Consistency problem, we would again uniformly at random pick $m$ triples of items $a,b,c$ and then add w.p. $(1-\eps)$ the constraint $ab|c$ if $c$ is separated first from $a,b$ in the ground-truth (rooted binary) tree.  For all problems, after getting $m$ (noisy) constraints in the analogous manner, our goal is to find a global solution that satisfies as many constraints as possible.

\paragraph{Positive Results:} Using our stochastic model we can escape worst-case impossibility results and for all 3 categories of problems, we present improved approximation algorithms.  At a high-level, we first construct a graph by encoding each of the local constraints on the items as a set of positive or negative edges between them. The graph captures the desired relationships and then, we find a good first split maximizing the ratio of satisfied over violated constraints by the cut. Naturally, our algorithm \MCS\ (see Algorithm~\ref{alg:decoding}) is based on variants of \MC\ on graphs with negative weights. An interesting building block in our analysis when solving for better Maximum Acyclic Subgraphs, is the directed \MC\  problem on graphs with negative weights which, to the best of our knowledge, hadn't been analyzed before. We note that for the triplets problem on trees, analogous  \MC\ heuristics had been successfully used before in experimental work for computational biology, however with no theoretical guarantees~\citep{snir2006using,snir2012quartet,snir2008quartets}. An exception is the work of~\cite{snir2012reconstructing}, where they focus only on the desired quartets problem, however their analysis is a special case of ours for when $\QF=\emptyset$ (i.e., the input contains no forbidden quartets). Our final approximations circumvent known hardness results for the case of rankings~\citep{guruswami2011beating} and our new hardness results for trees described in detail later in Section~\ref{sec:neg}.

\subsection{Better Approximations for \textsc{MAS}}

We start with \MAS\ as it is perhaps the easiest to describe (see also Algorithm~\ref{alg:decoding}):
\begin{theorem}
Given $m$ constraints generated according to our stochastic model on $n$ items, \MCS\ satisfies at least $(0.642-0.4285\eps)m$ on average, where $\eps$ is the fraction of erroneous comparisons. If moreover $m\ge \Omega(\log n)$, the result holds w.h.p.
\end{theorem}
\begin{remark}
For example, if the error parameter $\eps=0.1$, hence $10\%$ of the $m$ generated constraints are erroneous, we still satisfy $\approx 60\%$ of them, and we still beat the previous best $\tfrac12$-approximation together with the known hardness~\citep{guruswami2008beating}.
\end{remark}
Our general proof template has $5$ steps:  

%\vspace{-3mm}
\begin{itemize}
    \item Building a graph: For a sampled constraint $a<b$ indicating that $a$ should precede $b$ in the ranking, we add two directed edges:
\[
+1 \text{ directed from }a\to b, -1 \text{ directed from }b\to a
\]
    Since the problem has orientation, we define the weight of a directed cut $(S,\bS)$ as the sum of all (positively or negatively) weighted arcs going from $S$ to $\bS$ (and we ignore the arcs going from $\bS$ to $S$).
\vspace{-3mm}
    
    \item Cuts and constraints:  The goal of constructing the graph is to use information about its cuts and relate them to the pairwise constraints. Notice that a cut $(S,\bS)$ can either obey, disobey or leave unaffected the status of a $a<b$ constraint, depending on if $a$ or $b$ belongs to $S$ or $\bS$. Let $m_s, m_v$ denote the satisfied, violated constraints by the cut, respectively. The weight of any directed $(S,\bS)$ cut is thus:
    \begin{equation}\label{eq:comp}
      w(S,\bS)= m_s(S,\bS) - m_v(S,\bS)  
    \end{equation}

    as satisfied pairs $m_s$ (with $a\in S,b\in \bS$) contribute $+1$ and violated pairs $m_v$ (with $a\in \bS,b\in S$) contribute $-1$.
    
%\vspace{-3mm}

    \item Lower Bounding \MC: The constructed graph from the first step, is directed and has both positive and negative weights. Based on eq.~(\ref{eq:comp}), we should find a large cut in this graph as this translates to many satisfied constraints. In order to find the cut, we use a \MC\ variant that finds a cut comparable to the optimal max cut in graphs that are directed and contain both positive and negative weights. However, we cannot use the standard Goemans-Williamson algorithm and guarantees~\cite{goemans1995improved}, as the graph is directed with positive and negative weights. A new ingredient in our proof is a semidefinite programming relaxation and analysis for this variant that achieves:
    \begin{equation}\label{eq:guar}
    \EE(w(S,\bS)) \ge 0.857w(\OPT)-0.143\cdot W^-
    \end{equation}
    where $w(\OPT)$ is the weight of the optimum cut and $W^-$ is the total negative weight in the graph in absolute value. Based on the graph construction in this case, $W^-=m$ as every constraint contributed a $-1$ edge. We just note that the numerical values $0.143$ and $0.857$ sum to 1, and they just arise from the rounding scheme used to obtain an integral solution from the relaxation.

    \item Now that we have a lower bound for $w(S,\bS)$ based on the optimum cut, in order to conclude the algorithm's cut is large (and hence satisfies many constraints), we need to lower bound the optimum's cut weight $w(\OPT)$. To do this we consider the weight of a \textit{median} directed cut: the median cut is defined to be the one that assigns the first $n/2$ labels in the optimum ordering for \MAS, on one side of the cut, and the rest $n/2$ labels to the other side of the cut. Since the labels for the constraints according to our stochastic model were chosen at random, a counting argument implies that with high probability $\approx \tfrac12m$ of the generated constraints are satisfied by the median cut and hence also by \OPT. To see this, observe that for nearly half of the $a<b$ constraints, $a$ belongs to the first $n/2$ labels of the median cut, whereas $b$ belongs to the remaining $n/2$ labels. Since \OPT\ is by definition even better than the median cut, we get that it has a large cut value. If we wanted to be slightly more precise, we should say that due to errors in an $\eps$ fraction of the generated constraints, we actually lose a small $\eps$ fraction of the constraints (we defer details to Appendix~\ref{app:beating}) but this discounts the optimum cut only by a small amount. 
    
    \item Output of \MCS: Finally, we need to find a good permutation overall, not just a good top split. Our algorithm starts by finding an approximate \MC\ $(S,\bS)$ in $G$ and then proceeds by outputting a \textit{random} permutation on the items in $S$ and in $\bS$ and concatenating them. Finally, we can compute the overall value of \ALG\  (dropping the notation with $(S,\bS)$):
    \[\ALG = m_s+\tfrac12 m_u =\]
    \begin{equation}\label{eq:MAS12}
        = m_s+\tfrac12(m-m_s-m_v)=\tfrac12m+\tfrac12(w(S,\bS))
    \end{equation}
 where $m_u$ are the constraints that were unaffected by the  $(S,\bS)$ cut. By eq.~(\ref{eq:MAS12}), we already see that we get some advantage over the $\tfrac12m$ baseline which is optimal in the worst-case (and is achieved by a random permutation on all $n$ items).
\end{itemize}

\begin{remark}
A natural question is to attempt to use MaxCut repeatedly on each of the two generated parts of the first split. However analyzing the repeated MaxCut approach is not that simple, as once the first approximate MaxCut is performed, there is no randomness in the two generated subgraphs that we can exploit.  Analogous difficulties arise in dissimilarity-based and quartets-based hierarchical clustering~\cite{charikar2019hierarchical,snir2012reconstructing,ahmadian2019bisect}. Finally, we want to point out that such analyses are also known to be challenging from the literature on Random Forests for decision trees (e.g., \cite{scornet2015consistency}) where a similar (data-dependent) two-step analysis has been elusive.
\end{remark}

\subsection{Extensions to Other Problems} 

The same proof template as presented here can be modified to deal with the remaining problems: \BTW, \nBTW, forbidden and desired triplets, forbidden and desired quartets. As each of these constraints, involve $3$ or $4$ points, the construction and analyses become more involved. We present briefly the main modifications  for the graph construction (see Appendix~\ref{app:beating} for details).

For a \BTW\ constraint $\{a|b|c\}$, we add undirected edges: $+2$ for $(a,c)$ and $-1$ for $(b,a),(b,c)$. The edges capture that a cut violates the constraint if it separates $b$ from $a,c$. For a \nBTW\ constraint $\{ab|c\}$ indicating that $c$ should not be between $a,b$ in the final ordering, we add the following 3 undirected edges:+1 \text{for pairs }(c,a),(c,b)\text{ and } -2\text{ for the pair }(a,b). Recall, that for \BTW\ and \nBTW, the ultimate goal is to beat the factors $\tfrac13$ and $\tfrac23$ which are currently optimal in the worst-case:
\begin{theorem}
Given $m = \Omega(\log n)$ noisy constraints on $n$ items, variations of \MCS\ satisfy at least $(0.402-0.329\eps)m$  and $(0.845-0.329\eps)m$ constraints w.h.p. for \BTW\ and \nBTW, respectively, where $\eps$ is the fraction of erroneous constraints.
\end{theorem}

For Correlation Clustering, for each \CL\ constraint $ab$, we add a $+1$ for $(a,b)$, and for each \ML\ constraint $ab$, we add $-3.2735$ for edge $(a,b)$. The chosen numerical value $-3.2735$ depends on the current best 0.766-approximation for Correlation Clustering~\citep{swamy2004correlation} (see Appendix~\ref{app:beating}).

\begin{theorem}
Given $m = \Omega(\log n)$ noisy ``must-link/cannot-link'' constraints on $n$ items, \MCS\ (modified appropriately) satisfies at least $(0.8226-0.775\eps)m$ constraints w.h.p., where $\eps$ is the fraction of erroneous constraints.
\end{theorem}

Analogous theorems hold for the Triplets/Quartets consistency problems. Due to space constraints, we omit the statements but we refer the reader to Table~\ref{tab:result} for the final ratios and to Appendix~\ref{app:beating} for the proofs.

\begin{algorithm}[t]
\caption{Our \MCS\ template as instantiated for \MAS.}
\label{alg:decoding}
\begin{algorithmic}
        \State \textbf{Input:} $m$ pairwise constraints for \MAS.
        \State \textbf{1.} For each $a<b$ constraint, insert a $+1$ arc directed from $a\to b$ and another  arc with negative weight $-1$ directed from $b\to a$. Call the resulting graph $G$.
        \State \textbf{2.} Run our approximate \MC\  algorithm suitable for directed graphs with negative weights to get a first split $(S,\bS)$, satisfying~eq.~(\ref{eq:guar}).
        \State \textbf{3.} Construct a random permutation $\pi_1$ on the nodes in $S$ and a random permutation $\pi_2$ on the nodes in $\bS$. Let $\pi$ be the ranking obtained by concatenating $\pi_1$ and then $\pi_2$.
        \State \textbf{4.} Return $\pi$.
\end{algorithmic}
\end{algorithm}

\section{Hardness for CSPs on Trees}
\label{sec:neg}
\paragraph{Negative Results:}
  As mentioned, previous work~\citep{byrka2010new,jiang2001polynomial, bryant1997building, he2006inferring, steel1992complexity} tried to get better approximations for triplets/quartets consistency compared to trivial baselines. Recall, that the trivial baseline is to simply output a random tree (either rooted or unrooted depending on the problem).  In our paper, near optimal hardness of approximation results for the maximum desired/forbidden triplets/quartets consistency problems ($4$ problems in total) are presented shedding light to why, despite significant efforts from different communities, no improvement had been made for nearly thirty years.  As a consequence, we get the first tight hardness for an ordering problem on trees, thus extending the work of~\cite{guruswami2011beating} from orderings on the line to hierarchical clustering.
  
  Specifically, for maximizing forbidden triplets, we show that no polynomial time algorithm can achieve a constant better than $\tfrac23$-approximation. Similar to~\cite{guruswami2008beating,guruswami2011beating} this is assuming the Unique Games Conjecture, however for maximizing desired triplets, we show a threshold of $\tfrac23$, assuming P $\neq$ NP. The above also implies that forbidden triplets is approximation resistant as a random tree also achieves a $\tfrac23$ factor. In fact our hardness results for all 4 problems are stronger, as we show it's not possible to distinguish almost perfectly consistent inputs from inputs where the optimum solution achieves almost the same as a random solution. 
  
  Technically, in order to get the hardness results, we give algorithms to obtain permutations on the leaves of a tree, such that if the tree obeyed many triplet/quartet constraints, then the permutation would also obey a large fraction of them when viewed as appropriate ordering constraints. Specifically, we prove that under the \UGC, it is hard to approximate the  Forbidden Triplets Consistency problem better than a factor of $\tfrac23$, even in the unweighted case. %Our result is slightly stronger: it is hard to distinguish between two instances one of which is almost perfect (e.g., $99\%$ of constraints are consistent) and the other is far from perfect (e.g., $67\%$ of constraints are consistent). 
  
\begin{fact}
Let $K$ be the total number of triplet constraints in an instance of \BTW. For any $\epsilon>0$, it is UGC-hard to distinguish between \BTW\ instances of the following two cases:\\
\textbf{YES}: $val(\pi^*)\ge (1-\epsilon)K$, i.e. the optimal permutation satisfies almost all constraints. \\
\textbf{NO}: $val(\pi^*) \le (\tfrac13 +\epsilon)K$, i.e. the optimal permutation does not satisfy more than 1/3 fraction.
\end{fact}

Given the above fact from~\cite{guruswami2011beating}, we prove our $\tfrac23$-inapproximability result for Forbidden Triplets:
%\footnote{To get the tight 1/3 inapproximability result, one way would be to look at the OCSPs paper and check that in their YES instances, there is also a corresponding tree such that it satisfies almost all of the given constraints.}
\begin{theorem}\label{th:forbidden}
Let $K$ be the total number of the triplet constraints in an instance of Forbidden Triplets Consistency. For any $\delta>0$, it is UGC-hard to distinguish between the following two cases:\\
\textbf{YES}: $val(T^*)\ge (1-\delta)K$, i.e. the optimal tree satisfies almost all the triplet constraints. \\
\textbf{NO}: $val(T^*)\le (\tfrac23 +\delta)K$, i.e. the optimal tree does not satisfy more than $\tfrac23$ fraction of triplets.
\end{theorem}

%\vspace{-3mm}
\begin{proof}
Start with a YES instance of the \BTW\ problem with optimal permutation $\pi^*$ and $val(\pi^*)\ge (1-\epsilon)K$. Viewing each \BTW\ constraint $a|b|c$ as a forbidden triplet $ac|b$, we show how to construct a tree $T$ such that $val(T)\ge (1-\delta(\eps))K$. In fact, the construction is straightforward: simply assign the $n$ labels, in the order they appear in $\pi^*$, as the leaves of a caterpillar tree (every internal node has its left child being a leaf). Observe that this caterpillar tree satisfies: $val(T)\ge (1-\epsilon)K$. This is because if a \BTW\ constraint $a|b|c$ was obeyed by $\pi^*$, it will also be avoided (viewed as a forbidden triplet $ac|b$) by the caterpillar tree above: if $a$ appears first in the permutation then the caterpillar will avoid $ac|b$ as $a$ gets separated first, otherwise if $c$ appears first, then again the caterpillar tree will avoid $ac|b$ as $c$ gets separated first. 

The NO instance is more challenging. Start with a NO instance of the \BTW\ problem with optimal $\pi^*$ of value $val(\pi^*)\le (\tfrac13+\eps)K$. Viewing the \BTW\ constraints as forbidden triplets, we show that the optimum tree $T^*$  cannot achieve better than $>(2/3+2\epsilon)K$, because this would imply that $val(\pi^*)> (\tfrac13+\eps)K$, which is a contradiction. For this, assume that some tree $T$ scored a value $val(T)>(2/3+2\epsilon)K$. We will construct a permutation $\pi$ from the tree $T$ with value $val(\pi)>(1/3 +\epsilon)K$, a contradiction. Notice that there are forbidden triplets that may be avoided by the tree, yet obeyed by the permutation: for example for a forbidden triplet  $t=ac|b$, the tree $R$ that first removes $a$ and then splits $b,c$ will successfully avoid $t$, however the  permutation $acb$ can come from $R$ by projection, however $acb$ does not obey the \BTW\ constraint $a|b|c$. Hence directly projecting the leaves of $T$ onto a line may not satisfy $>(1/3+2\epsilon)K$, since every forbidden triplet $ac|b$ avoided by $T$, can be ordered by this projected permutation in a way that would not obey the corresponding \BTW\ constraint $a|b|c$.  However, just by randomly swapping each left and right child for every internal node in the tree before we do the projection to the permutation, would satisfy $1/2\cdot(2/3 +2\epsilon)K=(1/3+\epsilon)K$ number of constraints. To see this, note that with probability $\tfrac12$ a forbidden $ac|b$ avoided by $T$ will be mapped to the desired $abc$ (and not $acb$) or $cba$ (and not $cab$) ordering.

Finally, we get $val(\pi^*)\ge val(\pi)> (1/3+\epsilon)K$, a contradiction that we were given a NO instance. To conclude, $\tfrac23$-inapproximability follows from the gap of these two instances.
\end{proof}
For the Desired Triplets problem, the proof proceeds in a similar fashion. One main difference is that we prove hardness of $\tfrac23$ under P $\neq$ NP, without assuming \UGC. The reason is that we reduce from the \nBTW\ problem that is known to be approximation resistant, subject only to P $\neq$ NP. Of course, one open question is to close the gap between this $\tfrac23$ factor and the current best approximation of $\tfrac13$.

\begin{theorem}
Let $K$ be the total number of the triplet constraints in an instance of Desired Triplets Consistency. For any $\delta>0$, it is NP-hard to distinguish:\\
\textbf{YES}: $val(T^*)\ge (\tfrac12-\delta)K$\\
\textbf{NO}: $val(T^*)\le (\tfrac13 +\delta)K$
\end{theorem}

Switching to quartet problems, our reductions are more challenging. The first challenge is that constraints are on $4$ items so we need to resort to an ordering CSP of arity $4$, that we term \SEP. Next, trees are unrooted and we want to generate an ordering on their leaves. To do this we first root the tree at some internal node and then follow a similar strategy for randomly reordering their children. For desired quartets we show hardness of $\tfrac23$ and for forbidden quartets a hardness of $\tfrac89$ (see App. \ref{app:beating} for statements). Recall that the best approximations are $\tfrac13$ and $\tfrac23$ respectively, achieved by a random (unrooted) tree.

%\begin{theorem}
%Let $K$ be the total number of the quartet constraints in an instance of Desired Quartets Consistency. For any $\delta>0$, it is UGC-hard to distinguish between instances of the following two cases:\\
%\textbf{YES}: $val(T^*)\ge (1-\delta)K$, i.e. the optimal tree satisfies almost  all the quartet constraints. \\
%\textbf{NO}: $val(T^*)\le (\tfrac23 +\delta)K$, i.e. the optimal tree does not satisfy more than a $\tfrac23$ fraction of the quartet constraints.
%\end{theorem}

%\begin{theorem}
%Let $K$ be the total number of the quartet constraints in an instance of Forbidden Quartets Consistency. For any $\delta>0$, it is UGC-hard to distinguish between instances of the following two cases:\\
%\textbf{YES}: $val(T^*)\ge (1-\delta)K$, i.e. the optimal tree satisfies almost  all the quartet constraints. \\
%\textbf{NO}: $val(T^*)\le (\tfrac89 +\delta)K$, i.e. the optimal tree does not satisfy more than a $\tfrac89$ fraction of the quartet constraints.
%\end{theorem}

\begin{remark}
Note that our hardness results give optimal results when restricted to (rooted or unrooted) \textit{caterpillar} trees, an important tree family, where each internal node has at least one leaf as a child. 
\end{remark}

\section{Conclusion and Open Problems}

We studied ranking, correlation clustering and hierarchical clustering under qualitative constraints and we presented a simple algorithm based on \MC\ that is able to overcome known hardness results under a random model. We also provided the first tight hardness of approximation for CSPs on trees shedding light to basic problems in computational biology and extending previous results by~\cite{guruswami2011beating} from ordering CSPs to trees. 

In fact, we believe that a nice open question is to prove that the two most important families of CSPs on trees (triplets and quartets consistency) are approximation resistant. Here we showed this for the case of forbidden triplets. More generally, we conjecture that all non-trivial CSPs on trees are approximation resistant, implying that the inapproximability results of~\cite{guruswami2011beating} can be extended from linear orderings to trees.
%say about beyond worst case analyses of the above problems maybe?

\section*{Acknowledgments}
The authors would like to thank Alessandro Epasto for interesting discussions during early stages of this work.

\bibliographystyle{alpha}
\bibliography{main}

\appendix
\section{Omitted Proofs - Improved Approximations via \MCS}
\label{app:beating}

In this first section of the Appendix, we present the omitted details for our positive results. Specifically, we show how to overcome impossibility results (see also Appendix~\ref{app:hardness}) by going beyond the hardness of approximation thresholds $\rho$ for each of the problems considered in the paper. As noted, to escape the worst-case analysis, we will assume the input is given as a set of $m$ noisy constraints generated according to our stochastic model and the goal is to obtain a solution with strictly more than $\rho m$ satisfied constraints. 

Recall that in Table~\ref{tab:result}, only for the results on Correlation Clustering and on Triplets Consistency marked with an asterisk (*), we required a mild balancedness assumption. The assumption here on the balancedness of the ground truth partition or ground truth hierarchical clustering is used in our reduction, and specifically when analyzing our \MC\ approach. It is needed in order to ensure that based on our stochastic model, our \MC\ approach can find a large cut in the constructed graph which later translates into a large portion of satisfied constraints. 

\begin{assumption}
\label{ass:triplets}
For a tree with $n$ leaves, a split $(L,R)$ at an internal node is called \emph{balanced} if $|L|=cn, |R|=(1-c)n$ with $\tfrac13\le c\le\tfrac23$. We assume that in the optimum tree there exists one split that is balanced. Similarly, for a clustering on $n$ nodes, if there exists a partition of the clusters into two sides $(L,R)$ such that $|L|=cn, |R|=(1-c)n$ with $\tfrac13\le c\le\tfrac23$, we say the clustering is balanced.
\end{assumption}
 This is a reasonable assumption since hierarchical clusterings tend to be balanced and indeed recursive balanced cuts tend to recover good hierarchies~\citep{charikar2017approximate}. In essence, we exclude caterpillar trees or more generally highly skewed trees that are generated by always removing tiny pieces out of a giant component. Moreover, such assumptions are common in generative graph models for clustering, e.g., the Stochastic Block Model~\citep{mossel2012stochastic,abbe2015exact} and for hierarchical clustering, e.g., the Hierarchical Stochastic Block Model~\citep{lyzinski2016community,cohen2019hierarchical,ghoshdastidar2019foundations}, where we expect to see at least two large communities emerge. For example, recent generative models like the Hierarchical Stochastic Block Model in~\cite{ghoshdastidar2019foundations} satisfy the balancedness assumption with $c=\tfrac12$.

\subsection{Quartets Consistency from Noisy constraints}
Let $\mathcal{Q_F}, \mathcal{Q_D}$ be the set of forbidden and desired quartet constraints with sizes $|\QF|=m_1, |\QD|=m_2$  respectively. The total number of generated constraints according to our stochastic model is denoted by $m=m_1+m_2$. Out of those constraints, let $\eps_1,\eps_2$ denote the fraction of the erroneous forbidden and erroneous desired quartet constraints respectively. Our main theorem here is:

\begin{theorem}\label{th:quartets}
Given $m=m_1+m_2$ constraints as above on $n$ items, our algorithm \MCS\ satisfies at least $(0.425-0.261\eps_1)m_1 +  (0.672-0.296\eps_2)m_2 $ on average, where $\eps_1, \eps_2$ are as above. If moreover $m_1,m_2\ge \Omega(\log n)$, the result holds w.h.p.
\end{theorem}

For example, if the constraints are not erroneous (i.e., $\eps_1=\eps_2=0$), we satisfy 42.5\% of the desired quartets, while avoiding 67.2\% of the forbidden quartets, improving upon prior best approximations.

In order to prove Theorem~\ref{th:quartets}, we will require several intermediate lemmas and constructions.

Recall that forbidden quartets should be avoided, whereas desired quartets should be satisfied by the tree our algorithm finds. We use the following notation: Let $\QA(\ALG)$ denote the number of quartets $q\in\QF$ avoided, $\QF(\ALG)$ the number of quartets $q\in\QF$ \emph{not} avoided (of course, $\QF=m_2=\QA(\ALG)+\QF(\ALG)$) and $\QD(\ALG)$ the number of quartets $q\in \QD$ satisfied by the output phylogenetic tree. For the case of no errors $\eps_1=0, \eps_2=0$, the  best approximation under worst-case analysis is:

\[
\QD(\ALG)-\QF(\ALG) \ge \tfrac13(|\QD| - |\QF|) \iff \]\[\QD(\ALG)+\QA(\ALG) \ge \tfrac13 m_1 + \tfrac23 m_2
\]
In fact, the guarantees hold separately $\QD(\ALG)\ge  \tfrac13 |\QD|$ and $\QA(\ALG) \ge \tfrac23 |\QF|$ and are achieved either by a simple greedy algorithm or by a random tree \citep{he2006inferring}. Our goal is to find a tree beating the above guarantees, i.e., satisfying strictly more than $\tfrac13$ fraction of desired quartets and strictly more than $\tfrac23$ fraction of forbidden quartets. Our approach is based on extending a previous analysis from \cite{snir2006using} that only handled the case with $\QF=\emptyset$.

The end result of our algorithm \ALG, which is based on \MC, is a tree with the following guarantees:
\[\QD(\ALG)+\QA(\ALG) \ge\]
\begin{equation}\label{eq:beatq}
     \ge (0.425-0.261\eps_1)m_1 +  (0.672-0.296\eps_2)m_2 
\end{equation}

We start by instantiating our general algorithmic template in Algorithm~\ref{alg:decoding} to the case of the Quartets Consistency problem, and we describe the necessary changes for the appropriate graph construction below:

\paragraph{Graph Construction from constraints:} The goal here is to construct a graph encoding the qualitative information from the generated quartets so that a \MC\ subroutine can yield a reasonable first split of the output phylogenetic tree. Quartets $q \in \QF$ needs to be handled differently from quartets $q\in \QD$. For each forbidden $q=\{ab|cd\} \in \QF$ we add the following six $+$ or $-$ weighted edges:
\[
+2 \text{ for pairs }(a,b),(c,d)\text{ and }\]\[ -1\text{ for pairs }(a,c),(a,d),(b,c),(b,d)
\]

and for a $q=\{ab|cd\} \in \QD$ we add the following $+$ or $-$ edges:
\[
-2 \text{ for pairs }(a,b),(c,d) \text{ and }\]\[ +1\text{ for pairs }(a,c),(a,d),(b,c),(b,d)
\]
Let $G$ be the undirected weighed multigraph constructed from the constraints as above and let $(S,\bar{S})$ denote any graph cut into two parts. We say that a quartet $q=\{ab|cd\}\in \QF\cup\QD$ is \emph{unaffected} by the cut $(S,\bar{S}$ if all four labels $a,b,c,d$ end up in one of the two parts. For quartets whose endpoints are separated by the cut, we distinguish 3 cases: if one of the labels goes to one of the two parts while the remaining 3 labels go to the other part, we say that $q$ is \emph{postponed}. If precisely $a,b$ are contained in some part, while the other part contains precisely $c,d$, we say $q$ is \emph{obeyed}. In any other case, $q$ is $disobeyed$ (e.g., $a,c \in S$ and $b,d\in \bar{S}$ or the symmetric split $a,d\in S$ and $b,c\in \bar{S}$). The perhaps more natural terms $satisfied$ and $violated$ were not used as we deal both with desired and forbidden quartets and would be misleading when accounting for the maximization objective:
\begin{lemma}\label{lemma:cut}
The weight of any cut $(S,\bar{S})$ can be computed based on the status of the quartets as:
\[w(S,\bar{S}) = 2m^{\QF}_d(S,\bar S) - 4 m^{\QF}_o(S,\bar S)+\]
\begin{equation}\label{eq:cut}
    + 4 m^{\QD}_o(S,\bar S) - 2m^{\QD}_d(S,\bar S) 
\end{equation}

where $m^{\QF}_d, m^{\QD}_d$ is the number of disobeyed quartets by the cut that belong to $\QF, \QD$ respectively and similarly $m^{\QF}_o, m^{\QD}_o$ is the number of obeyed quartets from $\QF, \QD$ respectively.
\end{lemma}
\begin{proof}
Note that by our choice for the edge weights, if $q=\{ab|cd\}$ is postponed or unaffected by the cut $(S,\bar S$), its contribution to $w(S,\bar S)$ is 0 regardless of $q\in \QF$ or $q\in \QD$. Now, if a forbidden $q\in \QF$ is obeyed, that counts as a mistake and it decreases the weight of the cut by -4, whereas if it is disobeyed, that counts as a correct choice and it increases the weight of the cut by +2. Accordingly we compute the contribution for the desired quartets $q\in \QD$ as +4 if obeyed and -2 if disobeyed. Summing over all constraints gives us the lemma.
\end{proof}

The final step is to compute the overall quartets our algorithm had success on, relative to the sample sizes $m_1,m_2$:

\begin{lemma}\label{lem:quartet_MC} If $(S,\bS)$ is the first split of \ALG, the total number of quartets decomposed correctly is: 
\[
\ALG = \QA(\ALG) + \QD(\ALG)\ge \tfrac23m_1 + \tfrac13m_2+\tfrac16w(S,\bS)
\]
\end{lemma}
\begin{proof}
Let $m^{\QF}_p, m^{\QF}_u$ denote the number of postponed or unaffected by the cut forbidden quartets, and $m^{\QD}_p, m^{\QD}_u$ denote the number of postponed or unaffected by the cut desired quartets. Our algorithm first uses an approximation to \MC and then proceeds greedily (or randomly) to achieve the baseline guarantees by building a tree on $S$ and on $\bS$:
\[
\ALG \ge m^{\QF}_d(S,\bar S) + \tfrac23 (m^{\QF}_u(S,\bar S) + m^{\QF}_p(S,\bar S)) +\]\[+ m^{\QD}_o(S,\bar S)+\tfrac13(m^{\QD}_u(S,\bar S) + m^{\QD}_p(S,\bar S))
\]
For notation purposes, from now on we drop the parentheses $(S,\bS)$ from the terms since we always refer to the $(S,\bS)$ cut.  Observe that $m_1=m^{\QF}_d+m^{\QF}_o+m^{\QF}_u+m^{\QF}_p$ and similarly $m_2=m^{\QD}_d+m^{\QD}_o+m^{\QD}_u+m^{\QD}_p$. By substituting the terms for unaffected and postponed quartets we get:
\[\ALG\ge \]
\[
m^{\QF}_d  + \tfrac23 (m_1-m^{\QF}_d - m^{\QF}_o) + m^{\QD}_o+\tfrac13(m_2 -m^{\QD}_d - m^{\QD}_o) 
\]
\[
=\tfrac23m_1+\tfrac13m^{\QF}_d-\tfrac23m^{\QF}_o+\tfrac13m_2+\tfrac23m^{\QD}_o-\tfrac13m^{\QD}_d
\]
\[
=\tfrac23m_1+\tfrac13m_2+\tfrac16(2m^{\QF}_d- 4 m^{\QF}_o + 4 m^{\QD}_o - 2m^{\QD}_d)
\]
From equation~(\ref{eq:cut}), the last term is equal to the weight of the $(S,\bS)$ cut and this finishes the proof.
\end{proof}

Now we need to show that there is a good cut with high weight in the graph. Recall that the graph has positive and negative edges. For such graphs, the guarantee of the rounding algorithm of~\cite{goemans1995improved} is as follows:
\begin{fact}\label{fact}
For graphs with both positive and negative weights, one can efficiently find a cut $(S,\bS)$ with weight:
\[
w(S,\bS)\ge 0.878w(S^*,\bS^*)-0.122 W^-
\]
where $(S^*,\bS^*)$ is the optimum solution for \MC\ and $W^-$ is the absolute sum of all negative edge weights.
\end{fact}

The cut $(S,\bS)$ is produced in the same manner as in the standard Goemans-Williamson algorithm via random hyperplane rounding on their semidefinite relaxation for \MC.
We will use this fact to prove the following:
\begin{lemma}\label{lem:quartet_LB} The weight of the top split relative to the sizes of the quartet constraints is:
\[
w(S,\bS)\ge (0.03229-1.56\eps_1)m_1+(0.5525-1.56\eps_2)m_2 
\]
\end{lemma}
%Perhaps with exact 3-cut and negative cancelation we can get an improvement for when m=n^2 logn and further if m is dense then we can get a PTAS for forbidden quartets. There are extensions from the Snir paper that we don't focus here.
\begin{proof}
Observe that in the constructed graph, the total negative weight is $W^-= 4m_1+4m_2$ as each quartet adds a total negative weight of -4. In order to use Fact~\ref{fact}, we require a lower bound on the optimum value $w(S^*,\bS^*)$.

Notice that for any phylogenetic tree, since all 
internal vertices have three neighbors each (a trivalent
tree), we can always find an edge that induces a \emph{balanced} cut. For $n$ leaves, a cut $(L,R)$ is called balanced if $|L|=cn, |R|=(1-c)n$ with $\tfrac13\le c\le\tfrac23$. From our uniform generating model, recall that the number of quartet constraints the cut $(L,R)$ succeeds at is:
\[
\mathbb{E}(m^{\QF}_{d})=6c^2(1-c)^2(1-\eps_1)m_1\ \]\[ \mathbb{E}(m^{\QD}_{o})=6c^2(1-c)^2(1-\eps_2)m_2
\] and the number of constraints the cut $(L,R)$ fails at, due to the erroneous constraints is:
\[
\mathbb{E}(m^{\QF}_{o})=6c^2(1-c)^2\eps_1m_1,\ \ \  \mathbb{E}(m^{\QD}_{d})=6c^2(1-c)^2\eps_2m_2
\]
The quantity $c^2(1-c)^2$ with $\tfrac13\le c\le\tfrac23$ attains a minimum value of $\tfrac{4}{81}$ when $c=\tfrac13$; hence, from Lemma~\ref{lemma:cut}, the weight of the cut $(L,R)$ on the constructed graph is:
\[
w(L,R)\ge 2m^{\QF}_d(L,R) - 4 m^{\QF}_o(L,R) +\]\[+ 4 m^{\QD}_o(L,R) - 2m^{\QD}_d(L,R) 
\]
\begin{equation}\label{eq:optQ}
  \ge \tfrac{16}{27}(1-\eps_1)m_1 - \tfrac{32}{27}\eps_1m_1+\tfrac{32}{27}(1-\eps_2)m_2 - \tfrac{16}{27}\eps_2m_2  
\end{equation}

Of course the optimum cut has even larger weight than the specific balanced $(L,R)$ cut so: $w(S^*,\bS^*)\ge w(L,R)$. Substituting equation~(\ref{eq:optQ}) in Fact~\ref{fact} yields the lemma.
%Using the fact that there is a cut $(L,R)$ that is balanced, we show why our uniform sampling technique yields many quartets obeyed by the cut
\end{proof}

\begin{proof}[Proof of Theorem~\ref{th:quartets}]
From Lemma~\ref{lem:quartet_MC}, we have a lower bound on our algorithm's performance via the approximate max cut. Substituting the quantity $w(S,\bar{S})$ based on Lemma~\ref{lem:quartet_LB}, yields the theorem.
\end{proof}

From the above, notice that we can still beat the prior best baselines as long as the error rates are not too big ($\eps_1 \le 3.4\%$ and $\eps_2\le 35.4\%$).

\subsection{Triplets Consistency from Noisy constraints}

Here we show a similar approximation result but for Triplets. Let $\mathcal{T_F}, \mathcal{T_D}$ be the set of forbidden and desired triplet constraints with sizes $|\TF|=m_1, |\TD|=m_2$  respectively. The total number of generated constraints is denoted by $m=m_1+m_2$. Out of those constraints, let $\eps_1,\eps_2$ denote the fraction of the erroneous forbidden and erroneous desired triplet constraints respectively. 

\begin{theorem}\label{th:triplets}
Given $m=m_1+m_2$ constraints as above on $n$ items, our algorithm \MCS\ satisfies at least $(\tfrac23 +0.11378-0.5853\eps_1)m_1+(\tfrac13+0.30886-0.5853\eps_2)m_2 $ on average, where $\eps_1, \eps_2$ are as above. If moreover $m_1,m_2\ge \Omega(\log n)$, the result holds w.h.p.
\end{theorem}

For example, if the constraints are not erroneous (i.e., $\eps_1=\eps_2=0$), we satisfy $64\%$ of the desired triplets, while avoiding $78\%$ of the forbidden triplets. This latter ratio beats our worst-case inapproximability results for triplets (see also Appendix~\ref{app:hardness}).

The reason we stated the numerical values in this form is that the trivial baselines achieve ratios of $\tfrac23$ and $\tfrac13$ for $m_1$ and $m_2$ respectively.

Recall that forbidden triplets should be avoided, whereas desired triplets should be satisfied by the tree our algorithm finds. We use the following notation: Let $\TA(\ALG)$ denote the number of triplets $t\in\TF$ avoided, $\TF(\ALG)$ the number of triplets $t\in\QF$ \emph{not} avoided (of course, $\TF=m_2=\TA(\ALG)+\TF(\ALG)$) and $\TD(\ALG)$ the number of triplets $t\in \TD$ satisfied by the output rooted binary hierarchical tree. For the case of no errors $\eps_1=0, \eps_2=0$, the  best approximation under worst-case analysis is:
\[
\TD(\ALG)-\TF(\ALG) \ge \tfrac13(|\TD| - |\TF|) \iff\]\[ \TD(\ALG)+\TA(\ALG) \ge \tfrac13 m_1 + \tfrac23 m_2
\]
In fact, the guarantees hold separately $\TD(\ALG)\ge  \tfrac13 |\TD|$ and $\TA(\ALG) \ge \tfrac23 |\TF|$ and are achieved either by a simple greedy algorithm or by a random tree \citep{he2006inferring}. Our goal is to find a tree beating the above guarantees, i.e., satisfying strictly more than $\tfrac13$ fraction of desired triplets and strictly more than $\tfrac23$ fraction of forbidden triplets.

The end result of our algorithm \ALG, which is based on \MC, is a tree with the following guarantees:
\[\TD(\ALG)+\TA(\ALG) \ge\]
\begin{equation}%\label{eq:beatq}
     (\tfrac23 +0.11378-0.5853\eps_1)m_1+(\tfrac13+0.30886-0.5853\eps_2)m_2 
\end{equation}

We proceed by describing the necessary changes to be made in our algorithmic template in Algorithm~\ref{alg:decoding}, in order to handle the triplet constraints.

\paragraph{Graph Construction from constraints:} The goal here is to construct a graph encoding the qualitative information from the generated triplets so that a \MC\ subroutine can yield a reasonable first split of the output binary hierarchical tree. Triplets $t \in \TF$ need to be handled differently from triplets $t\in \TD$. For each forbidden $t=\{ab|c\} \in \TF$ we add the following 3 $+$ or $-$ undirected weighted edges:
\[
+2 \text{ for the pair }(a,b)\text{ and } -1\text{ for pairs }(c,a),(c,b)
\]

and for a $t=\{ab|c\} \in \TD$ we add the following $+$ or $-$ edges:
\[
-2 \text{ for the pair }(a,b)\text{ and } +1\text{ for pairs }(c,a),(c,b)
\]
Let $G$ be the undirected weighed multigraph constructed from the constraints as above and let $(S,\bar{S})$ denote any graph cut into two parts. We say that a triplet $t=\{ab|c\}\in \TF\cup\TD$ is \emph{unaffected} by the cut $(S,\bar{S}$ if all three labels $a,b,c$ end up in one of the two parts. For triplets whose endpoints are separated by the cut, we distinguish 2 cases: if precisely $a,b$ are contained in some part, while the other part contains precisely $c$, we say $t$ is \emph{obeyed}. In any other case, $t$ is $disobeyed$ (e.g., $a,c \in S$ and $b\in \bar{S}$ or the symmetric split $b,c\in S$ and $a\in \bar{S}$). The perhaps more natural terms $satisfied$ and $violated$ were not used as we deal both with desired and forbidden quartets and would be misleading when accounting for the maximization objective:
\begin{lemma}\label{lemma:Tcut}
The weight of any cut $(S,\bar{S})$ can be computed based on the status of the triplets as:
\[ w(S,\bar{S}) =\]
\begin{equation}\label{eq:Tcut}
   m^{\TF}_d(S,\bar S) - 2 m^{\TF}_o(S,\bar S) + 2 m^{\TD}_o(S,\bar S) - m^{\TD}_d(S,\bar S) 
\end{equation}

where $m^{\TF}_d, m^{\TD}_d$ is the number of disobeyed triplets by the cut that belong to $\TF, \TD$ respectively and similarly $m^{\TF}_o, m^{\TD}_o$ is the number of obeyed triplets from $\TF, \TD$ respectively.
\end{lemma}
\begin{proof}
Note that by our choice for the edge weights, if $t=\{ab|c\}$ is unaffected by the cut $(S,\bar S$]), its contribution to $w(S,\bar S)$ is 0 regardless of $t\in \TF$ or $t\in \TD$. Now, if a forbidden $t\in \TF$ is obeyed, that counts as a mistake and it decreases the weight of the cut by $-2$, whereas if it is disobeyed, that counts as a correct choice and it increases the weight of the cut by +1. Accordingly we compute the contribution for the desired triplets $t\in \TD$ as $+2$ if obeyed and $-1$ if disobeyed. Summing over all constraints gives us the lemma.
\end{proof}

The final step is to compute the overall quartets our algorithm had success on, relative to the sample sizes $m_1,m_2$:

\begin{lemma}\label{lem:triplet_MC} If $(S,\bS)$ is the first split of \ALG, the total number of triplets decomposed correctly is: 
\[
\ALG = \TA(\ALG) + \TD(\ALG)\ge \tfrac23m_1 + \tfrac13m_2+\tfrac13w(S,\bS)
\]
\end{lemma}
\begin{proof}
Let $m^{\TF}_u$ denote the number of  unaffected by the cut forbidden triplets, and $m^{\TD}_u$ denote the number of unaffected by the cut desired triplets. Our algorithm first uses an approximation to \MC\ and then proceeds greedily (or randomly) to achieve the baseline guarantees by building a tree on $S$ and on $\bS$:
\[
\ALG \ge m^{\TF}_d(S,\bar S) + \tfrac23 m^{\TF}_u(S,\bar S) + m^{\TD}_o(S,\bar S)+\tfrac13m^{\TD}_u(S,\bar S)
\]
For notation purposes, from now on we drop the parentheses $(S,\bS)$ from the terms since we always refer to the $(S,\bS)$ cut.  Observe that $m_1=m^{\TF}_d+m^{\TF}_o+m^{\TF}_u$ and similarly $m_2=m^{\TD}_d+m^{\TD}_o+m^{\TD}_u$. By substituting the terms for the unaffected triplets we get:
\[\ALG\ge \]
\[
m^{\TF}_d  + \tfrac23 (m_1-m^{\TF}_d - m^{\TF}_o) + m^{\TD}_o+\tfrac13(m_2 -m^{\TD}_d - m^{\TD}_o) 
\]
\[
=\tfrac23m_1+\tfrac13m^{\TF}_d-\tfrac23m^{\TF}_o+\tfrac13m_2+\tfrac23m^{\TD}_o-\tfrac13m^{\TD}_d
\]
\[
=\tfrac23m_1+\tfrac13m_2+\tfrac13(m^{\TF}_d - 2 m^{\TF}_o + 2 m^{\TD}_o - m^{\TD}_d) 
\]
From equation~(\ref{eq:Tcut}), the last term is equal to the weight of the $(S,\bS)$ cut and this finishes the proof.
\end{proof}

Now we can use again Fact~\ref{fact} to give a lower bound on the optimal cut. The cut $(S,\bS)$ is produced in the same manner as in the standard Goemans-Williamson algorithm via random hyperplane rounding on their semidefinite relaxation for \MC. We will use the fact to prove the following:

\begin{lemma}\label{lem:triplet_LB} The weight of the top split relative to the sizes of the triplet constraints is:
\[
w(S,\bS)\ge (0.3413-1.756\eps_1)m_1+(0.9266-1.756\eps_2)m_2 
\]
\end{lemma}
%Perhaps with exact 3-cut and negative cancelation we can get an improvement for when m=n^2 logn and further if m is dense then we can get a PTAS for forbidden quartets. There are extensions from the Snir paper that we don't focus here.
\begin{proof}
Observe that in the constructed graph, the total negative weight is $W^-= 2m_1+2m_2$ as each triplet adds a total negative weight of $-2$. In order to use Fact~\ref{fact}, we require a lower bound on the optimum value $w(S^*,\bS^*)$.

Here is the first time where we require Assumption~\ref{ass:triplets} about the balancedness of the ground truth tree. From our stochastic model, recall that the number of triplet constraints the cut $(L,R)$ succeeds at is:
\[
\mathbb{E}(m^{\TF}_{d})=(3c^2(1-c)+3c(1-c)^2)(1-\eps_1)m_1\]\[ \mathbb{E}(m^{\TD}_{o})=(3c^2(1-c)+ 3c(1-c)^2)(1-\eps_2)m_2
\] and the number of constraints the cut $(L,R)$ fails at, due to the erroneous constraints is:
\[
\mathbb{E}(m^{\TF}_{o})=(3c^2(1-c)+3c(1-c)^2)\eps_1m_1\]\[ \mathbb{E}(m^{\TD}_{d})=(3c^2(1-c)+3c(1-c)^2)\eps_2m_2
\]
The quantity $c^2(1-c)+c(1-c)^2$ with $\tfrac13\le c\le\tfrac23$ attains a minimum value of $\tfrac29$ when $c=\tfrac13$; hence, from Lemma~\ref{lemma:Tcut}, the expected weight of the cut $(L,R)$ on the constructed 
graph is:
\[w(L,R)\ge\]
\[
 m^{\TF}_d(L,R) - 2m^{\TF}_o(L,R) + 2 m^{\TD}_o(L,R) - m^{\TD}_d(L,R) 
\]
\begin{equation}\label{eq:opt}
  \ge \tfrac{2}{3}(1-\eps_1)m_1 - \tfrac{4}{3}\eps_1m_1+\tfrac{4}{3}(1-\eps_2)m_2 - \tfrac{2}{3}\eps_2m_2  
\end{equation}

Of course the optimum cut has even larger weight than the specific balanced $(L,R)$ cut so: $w(S^*,\bS^*)\ge w(L,R)$. Substituting equation~(\ref{eq:opt}) in Fact~\ref{fact} yields the lemma.
%Using the fact that there is a cut $(L,R)$ that is balanced, we show why our uniform sampling technique yields many quartets obeyed by the cut
\end{proof}

\begin{proof}[Proof of Theorem~\ref{th:triplets}]
From Lemma~\ref{lem:triplet_MC}, we have a lower bound on our algorithm's performance via the approximate max cut. Substituting the quantity $w(S,\bar{S})$ based on Lemma~\ref{lem:triplet_LB}, yields the theorem.
\end{proof}

From the above, notice that we beat the trivial baselines as we avoid $\approx 78\% >\tfrac23$ of the forbidden triplets and we satisfy $\approx 64\% > \tfrac13$ of the desired triplets.

\subsection{Rankings from Noisy constraints}
Here we will show how to beat the approximability thresholds for $3$ problems: \MAS, \BTW\ and \nBTW, even though our techniques can be extended to handle many other ordering problems and combinations of desired or forbidden ordering constraints.

\paragraph{Non-BTW:} The goal here is to beat the threshold of $\tfrac23$-approximation and as we will see a $0.84$-approximation is possible. The main difference again is on the way we construct the graph based on the generated triplet constraints. For a query $\{ab|c\}$ indicating that $c$ should not be between $a,b$ in the final ordering we add the following 3 undirected edges:
\[
+1 \text{ for pairs }(c,a),(c,b)\text{ and } -2\text{ for the pair }(a,b)
\]
The graph is as always constructed by inserting all these edges for each of the triplet constraints. We describe below the necessary changes for each of the steps of the template.

\begin{itemize}
    \item Contrary to previous ordering problems, here a cut into two pieces can either satisfy, postpone or leave unaffected the status of a triplet $\{ab|c\}$. The weight of the cut is:
    \[
w(S,\bS)=2m_s(S,\bS)-m_p(S,\bS)
\]
as a satisfied triplet contributes $+2$ in the objective ($(c,a), (c,b)$ are cut) while a postponed triplet contributes a total of $-1$ (labels $a$ and $b$ are separated). 

\item Our algorithm \ALG, starting with the $(S,\bS)$ cut and continuing randomly after that, scores a total objective (we drop the $(S,\bS)$ notation):
\[
\ALG= m_s+ \tfrac23m_u+\tfrac12m_p
\]
since even for postponed constraints there is still a $\tfrac12$ probability of correctly placing $c$ either first or last among the three labels. Substituting $m=m_s+m_p+m_u$ which is true for any cut:
\[\ALG= \]
\[
=m_s + \tfrac23(m-m_s-m_p)+\tfrac12m_p=\]\[=\tfrac23 m+\tfrac13 m_s-\tfrac16m_p=
\]
\begin{equation}\label{eq:nb1}
    =\tfrac23m+\tfrac16(2m_s-m_d)=\tfrac23m+\tfrac16w(S,\bS)
\end{equation}

\item The graph's total negative weight is $W^-=2m$ so the Goemans-Williamson guarantee is:

\begin{equation}\label{eq:nb2}
\EE(w(S,\bS)) = 0.878w(\OPT)-0.122\cdot 2m    
\end{equation}

We lower bound the weight $w(\OPT)$ by the weight of the \textit{median} cut: consider the median element $q$ in the unknown optimum permutation and then let one part of the split be the elements that precede $q$. Generally, in permutation problems, ensuring that a balanced cut with large cut value exists, is easier than problems on trees, as the median cut guarantees a 50-50 split. Since the labels for the constraints were chosen at random, a simple counting argument implies that in expectation $\tfrac34m$ (i.e., $3c^2(1-c)m+3(1-c)^2cm$ with $c=\tfrac12$) constraints are satisfied by the \OPT\ cut, so $w(\OPT)\ge 2\cdot\tfrac34(1-\eps)m -\tfrac34\eps m$ and we get $(0.845-0.329\eps)$-approximation by substituting in equation~(\ref{eq:nb2}) and then to (\ref{eq:nb1}). For example, even when $\approx 10\%$ are erroneous, we still get a $0.81$-approximation.
\end{itemize}

%\begin{remark}
%A less accurate type of constraints for this problem could answer $\{bc|a\}$ for a triplet that is (relatively) ordered as $abqc$ in the ordering, where $q$ is the median element. Notice that for this query, the cut \emph{postpones} its status so $w(\OPT)=2m_s(\OPT)-m_p(\OPT)$. Still, the same approach would yield a 0.81-approximation.
%\end{remark}

\paragraph{BTW:} The goal here is to beat the $\tfrac13$-approximation which is the current best for inconsistent instances of \BTW. We will get a $0.402$-approximation. If the instance is promised to be consistent, Makarychev~\cite{makarychev2012simple} gave an algorithm achieving $\tfrac12$-approximation. It is a divide and conquer algorithm that is simple and runs in linear time. A significantly slower algorithm based on semidefinite program with the same approximation guarantee was previously proposed by Chor and Sudan~\cite{chor1998geometric}. 

For a triplet $\{a|b|c\}$ indicating that $b$ should be between $a$ and $c$ in the ordering we construct a graph with undirected edges:
\[
+2 \text{ for the pair }(a,c)\text{ and } -1\text{ for the pairs }(b,a),(b,c)
\]
The edges try to capture that a cut violates the constraint if it separates $b$ from $a,c$. We give our main steps:
\begin{itemize}
    \item Contrary to \nBTW, a cut into two pieces here can either violate, postpone or leave unaffected the status of the triplet $\{a|b|c\}$. The weight of a cut is:
    \[
    w(S,\bS)=m_p(S,\bS)-2m_v(S,\bS)
    \]
    as violated triplets contribute $-2$ and postponed triplets $+1$.
    \item Crucially, a postponed by the cut triplet, can still be satisfied with probability $\tfrac12$ and this gives us the advantage:
    \[\ALG = \tfrac13m_u+\tfrac12m_p=\]
    \[ =\tfrac13(m-m_p-m_v)+\tfrac12m_p=\]
    \begin{equation}\label{eq:b1}
        =\tfrac13m+\tfrac16(m_p-2m_v)=\tfrac13m+\tfrac16w(S,\bS)
    \end{equation}

    \item Again the graph's total negative weight is $W^-=2m$ so the Goemans-Williamson guarantee is:

\begin{equation}\label{eq:b2}
\EE(w(S,\bS)) = 0.878w(\OPT)-0.122\cdot 2m    
\end{equation}

As before, we lower bound the weight $w(\OPT)$ by the weight of the \textit{median} cut. Since the labels for the constraints were chosen at random, a simple counting argument implies that in expectation $\tfrac34m$ (i.e., $3c^2(1-c)m+3(1-c)^2cm$ with $c=\tfrac12$) constraints are postponed by the \OPT\ cut, so $w(\OPT)\ge \tfrac34(1-\eps)m-2\cdot\tfrac34\eps m$ and we get a $(0.402- 0.329\eps)$-approximation by substituting in equation~(\ref{eq:b2}) and then to (\ref{eq:b1}). For an error rate of $\approx 10\%$ we still get $\ge 0.369$-approximation, which is better than $\tfrac13$. 
\end{itemize}

\paragraph{MAS:} The goal here is to beat the trivial $\tfrac12$-approximation achieved by an arbitrary or its reversed (or a random) ordering. We will indeed be able to achieve a 0.642-approximation:

\begin{theorem}\label{th:maxDICUT}
Given $m$ constraints generated according to our stochastic model on $n$ items, \MCS\ satisfies at least $(0.642-0.4285\eps)m$ on average, where $\eps$ is the fraction of erroneous comparisons. If moreover $m\ge \Omega(\log n)$, the result holds w.h.p.
\end{theorem}

 The constraints here are on pairs of labels, e.g., $a<b$. Contrary to \BTW\ and \nBTW\ where the constructed graph and cuts were undirected, \MAS\ is \textit{orientated} in the sense that it matters which side of the cut the labels end up at. This introduces the first challenge since we have to solve approximate \MC\ in directed graphs with negative weights. For a query $a<b$ indicating that $a$ should precede $b$ in the ranking, we add two directed edges:
\[
+1 \text{ directed from }a\to b \text{ and another  arc with}\]\[ \text{\ negative weight } -1 \text{ directed from }b\to a
\]

Here the weight of a directed cut $(S,\bS)$ is the sum of all (positively or negatively) weighted arcs going from $S$ to $\bS$ (and we ignore the arcs going from $\bS$ to $S$). Here a cut can either satisfy, violate or leave unaffected the status of a query and there are no postponed constraints as they only involve two labels. We describe our steps:
\begin{itemize}
    \item It is easy to see that the weight of any directed $(S,\bS)$ cut is:
    \[
    w(S,\bS)= m_s(S,\bS) - m_v(S,\bS)
    \]
    as satisfied pairs contribute $+1$ and violated pairs contribute $-1$.
    \item Again we can compute the value of \ALG (dropping the notation with $(S,\bS)$):
    \[\ALG = m_s+\tfrac12 m_u =\]
    \begin{equation}\label{eq:MAS1}
        = m_s+\tfrac12(m-m_s-m_v)=\tfrac12m+\tfrac12(w(S,\bS))
    \end{equation}
    
    \item Again the graph's total negative weight is $W^-=m$. However now that the graph is directed and with negative weights, we cannot use the Goemans-Williamson guarantee. 
    A new ingredient in our proof is an SDP relaxation and rounding scheme that achieves:
    \begin{equation}\label{eq:DICUT}
    \EE(w(S,\bS)) = 0.857w(\OPT)-0.143\cdot W^-  
    \end{equation}

    \item Continuing as before, we will lower bound the weight $w(\OPT)$ by the weight of the \textit{median} directed cut (as noted in the main body, this cut simply separates  the first half of the items in the optimal ordering from the last half). Since the labels for the constraints were chosen at random, a simple counting argument implies that in expectation $\tfrac12m$ (i.e., $2c(1-c)m$ with $c=\tfrac12$) constraints are satisfied by the \OPT\ cut, so $w(\OPT)\ge \tfrac12m(1-\eps)-\tfrac12\eps m$ due to errors in $\eps$ fraction of the constraints. 
 \end{itemize}

    \begin{proof}[Proof of Theorem~\ref{th:maxDICUT}]
    Given the above observations, in order to get a $(0.642-0.4285\eps)$-approximation, we first substitute in equation~(\ref{eq:DICUT}) the lower bound we got for w(\OPT), and then we substitute $w(S,\bS)$ to (\ref{eq:MAS1}).
    \end{proof}
    
     For example, if $ 10\% $ of the constraints are erroneous we still satisfy $ \approx 60\%$ of all constraints, beating the worst-case inapproximability results of~\cite{guruswami2011beating}.

\subsubsection{Directed \MC\ with negative weights}
Here we proceed by proving an important ingredient in our proof relating to finding directed cuts in graphs with negative weights.

In the seminal paper by~\cite{goemans1995improved}, they show how directed \MC\ can be solved approximately on directed graphs with non-negative weights. They  used the following semidefinite programming relaxation where $A$ denotes the arcs of the graph and $V$ the vertices ($|V|=n$):
\[
\text{maximize } \tfrac14\sum_{(i,j)\in A} w_{ij}(1+v_0v_i-v_0v_j-v_iv_j)
\]
\[
\text{subject to: } ||v_i||^2=1, v_i \in \mathbb{R}^{n+1}, \forall i\in V\cup0 
\]
Notice the special role of the vector $v_0$, which is used to break the symmetry indicating that we want to maximize edges going from left to right where left is the side in which $v_0$ belongs to. Observe that in an integral $\{\pm 1\}$ solution if vertex $i$ is on the same side with $v_0$ and $j$ is on the other side then $(1+v_0v_i-v_0v_j-v_iv_j)=4$ that's why we chose the coefficient $\tfrac14$ in front of the summation. Also note that due to the symmetry if instead of $v_j$ we set $-v_j$ the relaxation won't change so we can instead think of:
\[
\text{maximize } \tfrac14\sum_{(i,j)\in A} w_{ij}(1+v_0v_i+v_0v_j+v_iv_j)
\]
This will just simplify some trigonometric expressions later.

In this subsection we will prove a bound on the weight of the cut for directed graphs with positive and negative edge weights. The bound we will be able to show is:
\begin{equation}
    \EE(w(S,\bS)) = 0.857w(\OPT)-0.143\cdot W^-  
\end{equation}
where $W^-$ denotes the total weight in absolute value of all negative edges. Notice that if no negative weights are present ($W^-=0$ then we almost recover the Goemans-Williamson $0.878$ coefficient. The above bound follows from the following theorem by rearranging terms:

\begin{theorem}
Let $W^-=\sum_{(i,j)\in A}|w^-_{ij}|$ where $x^-=\min (0,x)$. Then we can efficiently find a cut $(L,R)$ such that:
\[
\EE(w(L,R))+W^-\ge 0.857\left( w(\OPT)+W^- \right)
\]
where \OPT\ denotes the optimum directed cut in the graph.
\end{theorem}
\begin{proof}
Let \SDP\ denote the optimal \SDP\ value which is larger than $w(\OPT)$ since we relaxed the problem. We will show the above bound where  $w(\OPT)$ is replaced by \SDP. We need to rewrite the \SDP\ relaxation to incorporate the $W^-$ term and then we need to compute the probabilities an edge $(i,j)\in A$ participates or does not participate in the cut and how it compares to the contribution in the SDP relaxation. The probability an edge does not participate in the cut is needed here because negatively weighted edges exist, which could potentially decrease the value of the cut. Separating the positive and negative weights $(A=A^+\cup A^-)$ and rewriting the \SDP\ ($\theta_{ij}$ denotes the angle between $v_i,v_j$):
\[
\tfrac14\sum_{(i,j)\in A} w_{ij}(1+v_0v_i+v_0v_j+v_iv_j) + W^- =
\]
\[
=\tfrac14\sum_{(i,j)\in A^+} w_{ij}(1+v_0v_i+v_0v_j+v_iv_j) +\]\[ +\tfrac14\sum_{(i,j)\in A^-} |w_{ij}|\left(4-(1+v_0v_i+v_0v_j+v_iv_j)\right)=   
\]
\[
=\tfrac14\sum_{(i,j)\in A^+} w_{ij}(1+\cos\theta_{0i}+\cos\theta_{0j}+\cos\theta_{ij}) + \]\[+\tfrac14\sum_{(i,j)\in A^-} |w_{ij}|\left(3-\cos\theta_{0i}-\cos\theta_{0j}-\cos\theta_{ij}\right)  
\]

For the rounding algorithm we can use the standard Goemans Williamson rounding although this will only guarantee a sub-optimal coefficient of $0.796$ instead of $0.857$ in Equation~(\ref{eq:DICUT}). We will show later how a non-standard but better rounding scheme by~\cite{feige1995approximating} gives us the desired $0.857$ factor.

Let $r$ be a vector drawn uniformly from the unit sphere. Let's evaluate the contribution of a positive arc $(i,j)\in A^+$ to the quantity $\EE(w(L,R))+W^-$:
\[
\tfrac14 w_{ij}\left(4\cdot \mathbf{Pr}[\sgn(v_ir) = \sgn(v_jr) = \sgn(v_0 r)]\right)=\]\[=w_{ij}\mathbf{Pr}[\sgn(v_ir) = \sgn(v_jr) = \sgn(v_0 r)]
\]

For a negative arc $(i,j)\in A^-$, the contribution to the quantity $\EE(w(L,R))+W^-$ is:
\[
-\tfrac14 |w_{ij}|\left(4 \mathbf{Pr}[\sgn(v_ir) = \sgn(v_jr) = \sgn(v_0 r)]\right)+|w_{ij}|=
\]
\[
=|w_{ij}|(1-\mathbf{Pr}[\sgn(v_ir) = \sgn(v_jr) = \sgn(v_0 r)])
\]

Finally, if we can manage to lower bound $\mathbf{Pr}[\sgn(v_ir) = \sgn(v_jr) = \sgn(v_0 r)]$ by $(1+\cos\theta_{0i}+\cos\theta_{0j}+\cos\theta_{ij})$ and \emph{simultaneously} lower bound $(1-\mathbf{Pr}[\sgn(v_ir) = \sgn(v_jr) = \sgn(v_0 r)])$ by $(3-\cos\theta_{0i}-\cos\theta_{0j}-\cos\theta_{ij})$ we will have finished as the final result will follow by linearity of expectations. This can indeed be done using some trigonometric facts and the symmetry of spherical geometry:
\begin{fact} Let $r$ be chosen uniformly at random from the unit sphere. Then for any three vectors $v_i,v_j,v_0$ in the unit sphere:
\[
\mathbf{Pr}[\sgn(v_ir) = \sgn(v_jr) = \sgn(v_0 r)]=\]\[=1-\tfrac{1}{2\pi}(\theta_{ij}+\theta_{j0}+\theta_{i0})\ge\]
\[\ge 0.796\cdot\tfrac14(1+\cos\theta_{0i}+\cos\theta_{0j}+\cos\theta_{ij})
\]
and also:
\[
1-\mathbf{Pr}[\sgn(v_ir) = \sgn(v_jr) = \sgn(v_0 r)]=\]\[=\tfrac{1}{2\pi}(\theta_{ij}+\theta_{j0}+\theta_{i0})\ge\]
\[\ge 0.878\cdot\tfrac14(3-\cos\theta_{0i}-\cos\theta_{0j}-\cos\theta_{ij})
\]
\end{fact}

Putting it all together and using linearity of expectations we have shown:
\[
\EE(w(L,R))+W^-\ge 0.796\left( \SDP+W^- \right)\ge\]\[\ge 0.796\left( w(\OPT)+W^- \right)
\]

As we shall see next the first inequality is the one that determines the approximation coefficient. The above proves so far that 0.796 is possible. However there exists a more complicated rounding scheme which does not choose $r$ uniformly at random. It was developed in the context of \textsc{Max-2-SAT} problem by Feige and Goemans and their main idea behind their improvement is to take advantage of
the special role of $v_0$. They crucially use $v_0$: they map each $v_i$ to another vector $w_i$ that depends both on $v_i$ and on $v_0$, and only then they proceed with the Goemans-Williamson rounding algorithm. Specifically,$w_i$ is coplanar with $v_0$, on the
same side of $v_0$ as $v_i$ is, and forms an angle
with $v_0$ equal to $f(\theta_{i0})$. By choosing the function $f$ to be:
\[
f_{1/2}(\theta)=\tfrac12\theta+\tfrac12(\tfrac{\pi}{2}(1-\cos\theta))
\]
they report that they get a coefficient $0.857$ for the first inequality above (instead of $0.796$) and simultaneously a coefficient $0.9249$ for the second inequality (instead of $0.878$). Using again linearity of expectation, this implies our theorem:
\[
\EE(w(L,R))+W^-\ge 0.857\left( w(\OPT)+W^- \right)
\]
\end{proof}

\subsection{Correlation Clustering from Noisy Constraints}

The last of the proofs for the positive results will be for the Correlation Clustering problem, following the same ideas as in the proofs above. In correlation clustering, the information comes as \textsc{Must-Link} ($ab$) or \textsc{Cannot-Link}  ($a|b$) constraints indicating if two labels should be in the same or in different parts of an optimal partition. The current best algorithm is a $0.7666$-approximation by~\cite{swamy2004correlation} and here we improve under our stochastic model for the input constraints. We achieve a $(0.8226-0.775\eps)$-approximation.

\[
\ALG = m_s + 0.766m_u=m_s+0.766(m-m_s-m_v)=\]\[=0.766m+0.234(m_s-3.2735m_v)
\]

We construct an undirected graph where for every \CL\ constraint $ab$ we add a $+1$ edge between $a,b$ and for every \ML\ constraint $ab$ we add an edge $a,b$ now with negative weight $-3.2735$. 

\[
w(S,\bS)= m_s(S,\bS) - 3.2735m_v(S,\bS)
\]
Hence:
\[
\ALG = 0.766m+0.234w(S,\bS)
\]
Assuming that the largest cluster in the optimum partition has size at most $\tfrac n2$, our stochastic model will generate at least $\tfrac m2$ \CL\ constraints by a simple counting argument. This is in expectation, but of course using a standard large deviation Chernoff bound, all our claims in this paper can be made to hold with high probability. This also implies that the total number of \ML\ constraints is at most $\tfrac m2$. Thus, once again using \MC\ for the first split:
\[
\EE(w(S,\bS)) = 0.878w(\OPT)-0.122\cdot 3.2735\cdot \tfrac m2    
\]
An easy lower bound for the value of the \OPT\ cut is: $w(\OPT)\ge \tfrac m2$ hence we obtain a 0.8226-approximation.

\section{Hardness via Ordering CSPs}
\label{app:hardness}
In this part of the Appendix, we present our hardness of approximation results for the constraint satisfaction problems on trees, extending in some cases the inapproximability results of~\cite{guruswami2011beating,austrin2013np} from linear orderings to trees.

\subsection{Hardness for Rooted Triplets Consistency} 
We prove that under the \UGC, it is hard to approximate the  Desired Triplets Consistency problem better than a factor of $\tfrac23$, even in the unweighted case. Notice that the current best approximation is $\tfrac13$ achieved by a random tree (or a simple greedy algorithm). In fact our result is slightly stronger: it is hard to distinguish between two instances one of which is almost perfect (e.g., $99\%$ of constraints are consistent) and the other is far from perfect (e.g., $67\%$ of constraints are consistent). We base our hardness result on the following theorem by~\cite{austrin2013np} about the \textit{Non-Betweeness} problem and its $\tfrac23$-inapproximability:
\begin{fact}
Let $K$ be the total number of triplet constraints in an instance of \nBTW. For any $\epsilon>0$, it is NP-hard to distinguish between \nBTW\ instances of the following two cases:\\
\textbf{YES}: $val(\pi^*)\ge (1-\epsilon)K$, i.e. the optimal permutation satisfies almost all constraints. \\
\textbf{NO}: $val(\pi^*) \le (\tfrac23 +\epsilon)K$, i.e. the optimal permutation does not satisfy more than 2/3 fraction of the constraints.
\end{fact}

Given the above fact, we prove our $\tfrac23$-inapproximability result for Triplets Consistency:
%\footnote{To get the tight 1/3 inapproximability result, one way would be to look at the OCSPs paper and check that in their YES instances, there is also a corresponding tree such that it satisfies almost all of the given constraints.}
\begin{theorem}
Let $K$ be the total number of the triplet constraints in an instance of Desired Triplets Consistency. For any $\delta>0$, it is NP-hard to distinguish between instances of the following two cases:\\
\textbf{YES}: $val(T^*)\ge (\tfrac12-\delta)K$, i.e. the optimal tree satisfies almost half of all the triplet constraints. \\
\textbf{NO}: $val(T^*)\le (\tfrac13 +\delta)K$, i.e. the optimal tree does not satisfy more than $\tfrac13$ fraction of the triplet constraints.
\end{theorem}

Then, our $\tfrac23$-inapproximability result follows directly from the gap of these instances: $\tfrac13/\tfrac12 = \tfrac23$.

\begin{proof}
Start with a YES instance of the \nBTW\ problem with optimal permutation $\pi^*$ and $val(\pi^*)\ge (1-\epsilon)K$. Viewing each \nBTW\ constraint as a desired triplet, we show how to construct a tree $T$ such that $val(T)\ge (\tfrac12-\delta(\eps))K$. In fact, the construction is straightforward: simply assign the $n$ labels, either in the order they appear in $\pi^*$ or reversed, as the leaves of a caterpillar tree (every internal node has at least one child that is a leaf). Observe that this tree satisfies:
\[
val(T)\ge (1-\epsilon)K/2
\]
This is because if a \nBTW\ constraint $ab|c$ was obeyed by $\pi^*$, it will also be obeyed by one of the two caterpillar trees above: if $c$ appears first in the permutation then the former caterpillar will obey $ab|c$ as $c$ gets separated first, otherwise if $c$ appears last, then the reversed caterpillar tree will obey $ab|c$. Here the $\tfrac12$ factor is tight, since for example, the two \nBTW\ constraints $ab|c$ and $bc|a$ are both satisfied by the ordering $abc$, but when viewed as desired triplets, they cannot both be satisfied by a tree. 

The NO instance is slightly more challenging. 
Start with a NO instance of the \nBTW\ problem with optimal $\pi^*$ of value $val(\pi^*)\le (\tfrac23+\eps)K$. Viewing the \nBTW\ constraints as desired triplets, we show that the optimum tree $T^*$  cannot achieve better than $>(1/3+2\epsilon)K$, because this would imply that $val(\pi^*)> (\tfrac23+\eps)K$, which is a contradiction.

For this, assume that some tree $T$ scored a value $val(T)>(1/3+2\epsilon)K$. We will construct a permutation $\pi$ from the tree $T$ with value $val(\pi)>(2/3 +\epsilon)K$. Observe that directly projecting the leaves of $T$ onto a line (just outputting the $n$ leaves from left to right as they appear in the tree) would already satisfy $>(1/3+2\epsilon)K$, since every desired triplet $ab|c$ obeyed by the tree, will also be obeyed (as a \nBTW\ constraint) by $\pi$ as $c$ will either be first or last among the three labels $a,b,c$. 

Moreover, there are potentially desired triplet constraints that are disobeyed by the tree $T$, yet obeyed by the permutation.  We know that the number of remaining constraints is $K-(1/3+2\epsilon)K = (2/3 - 2\epsilon)K$. By randomly swapping each left and right child in the tree $T$ before we do the projection to the permutation $\pi$, will actually lead to an excess of $1/2\cdot(2/3 - 2\epsilon)K=(1/3-\epsilon)K$ number of \nBTW\ constraints. To see this notice that for every triplet that is disobeyed in the tree, there is a $\tfrac12$ probability that it becomes obeyed in the permutation. Summing up, we get $val(\pi) > (1/3+2\epsilon)K+(1/3-\epsilon)K > (2/3+\epsilon)K \implies val(\pi^*)\ge val(\pi)> (2/3+\epsilon)K$, a contradiction.
\end{proof}

\subsection{Hardness for Forbidden Triplets: Random is Optimal}

We prove that under the \UGC, it is hard to approximate the  Forbidden Triplets Consistency problem better than a factor of $\tfrac23$, even in the unweighted case. Notice that the current best approximation is in fact $\tfrac23$ achieved by a random tree (or a simple greedy algorithm), hence we settle the computational complexity of the problem. Our result is slightly stronger: it is hard to distinguish between two instances one of which is almost perfect (e.g., $99\%$ of constraints are consistent) and the other is far from perfect (e.g., $67\%$ of constraints are consistent). We base our hardness result on the following theorem by~\cite{guruswami2011beating} about the \BTW\ problem and its $\tfrac13$-inapproximability:
\begin{fact}
Let $K$ be the total number of triplet constraints in an instance of \BTW. For any $\epsilon>0$, it is UGC-hard to distinguish between \BTW\ instances of the following two cases:\\
\textbf{YES}: $val(\pi^*)\ge (1-\epsilon)K$, i.e. the optimal permutation satisfies almost all constraints. \\
\textbf{NO}: $val(\pi^*) \le (\tfrac13 +\epsilon)K$, i.e. the optimal permutation does not satisfy more than 1/3 fraction of the constraints.
\end{fact}

Given the above fact, we prove our $\tfrac23$-inapproximability result for Forbidden Triplets Consistency:
%\footnote{To get the tight 1/3 inapproximability result, one way would be to look at the OCSPs paper and check that in their YES instances, there is also a corresponding tree such that it satisfies almost all of the given constraints.}
\begin{theorem}
Let $K$ be the total number of the triplet constraints in an instance of Forbidden Triplets Consistency. For any $\delta>0$, it is UGC-hard to distinguish between instances of the following two cases:\\
\textbf{YES}: $val(T^*)\ge (1-\delta)K$, i.e. the optimal tree satisfies almost half of all the triplet constraints. \\
\textbf{NO}: $val(T^*)\le (\tfrac23 +\delta)K$, i.e. the optimal tree does not satisfy more than $\tfrac23$ fraction of the triplet constraints.
\end{theorem}

Then, our $\tfrac23$-inapproximability result follows directly from the gap of these instances: $\tfrac23/1 = \tfrac23$. 

\begin{proof}
Start with a YES instance of the \BTW\ problem with optimal permutation $\pi^*$ and $val(\pi^*)\ge (1-\epsilon)K$. Viewing each \BTW\ constraint $a|b|c$ as a forbidden triplet $ac|b$, we show how to construct a tree $T$ such that $val(T)\ge (\tfrac1-\delta(\eps))K$. In fact, the construction is straightforward: simply assign the $n$ labels, in the order they appear in $\pi^*$, as the leaves of a caterpillar tree (every internal node has its left child being a leaf). Observe that this caterpillar tree satisfies:
\[
val(T)\ge (1-\epsilon)K
\]
This is because if a \BTW\ constraint $a|b|c$ was obeyed by $\pi^*$, it will also be avoided (viewed as a forbidden triplet $ac|b$) by the caterpillar tree above: if $a$ appears first in the permutation then the caterpillar will avoid $ac|b$ as $a$ gets separated first, otherwise if $c$ appears first, then again the caterpillar tree will avoid $ac|b$ as $c$ gets separated first. 

The NO instance is slightly more challenging. 
Start with a NO instance of the \BTW\ problem with optimal $\pi^*$ of value $val(\pi^*)\le (\tfrac13+\eps)K$. Viewing the \BTW\ constraints as forbidden triplets, we show that the optimum tree $T^*$  cannot achieve better than $>(2/3+2\epsilon)K$, because this would imply that $val(\pi^*)> (\tfrac13+\eps)K$, which is a contradiction.

For this, assume that some tree $T$ scored a value $val(T)>(2/3+2\epsilon)K$. We will construct a permutation $\pi$ from the tree $T$ with value $val(\pi)>(1/3 +\epsilon)K$, a contradiction. Notice that there are forbidden triplets that may be avoided by the tree, yet obeyed by the permutation: for example for a forbidden triplet  $t=ac|b$, the tree $R$ that first removes $a$ and then splits $b,c$ will successfully avoid $t$, however the  permutation $acb$ can come from $R$ by projection, however $acb$ do not obey the \BTW\ constraint $a|b|c$.

Hence directly projecting the leaves of $T$ onto a line may not satisfy $>(1/3+2\epsilon)K$, since every forbidden triplet $ac|b$ avoided by $T$, can be ordered by this projected permutation in a way that would not obey the corresponding \BTW\ constraint $a|b|c$.  

However, just by randomly swapping each left and right child for every internal node in the tree before we do the projection to the permutation, would satisfy $1/2\cdot(2/3 +2\epsilon)K=(1/3+\epsilon)K$ number of constraints. To see this, note that with probability $\tfrac12$ a forbidden $ac|b$ avoided by $T$ will be mapped to the desired $abc$ (and not $acb$) or $cba$ (and not $cab$) ordering.

Finally, we get $val(\pi) > (1/3+\epsilon)K \implies val(\pi^*)\ge val(\pi)> (1/3+\epsilon)K$, a contradiction that we were given a NO instance.
\end{proof}

\subsection{Hardness for Desired Quartets Consistency}

The main result in this section is that for the desired quartets problem, one cannot do better than $\tfrac23$-approximation. Notice that a random unrooted tree achieve $\tfrac13$-approximation which is currently the best known algorithm.

To prove our results, we make use of a consequence from the results in~\cite{guruswami2011beating} for orderings CSPs of arity 4. Specifically, we define the following problem, which we call \textsc{4-Separatedness}.

\begin{definition}
For an ordering problem, a \textsc{4-Separatedness} constraint $\{ab|cd\}$ specifies that both elements $a,b$ should precede $c,d$ or that both $c,d$ should precede $a,b$ in the output ordering (e.g., $badc$, but not $acbd$). No constraints are placed on the relative ordering between $a,b$ or on the ordering between $c,d$.
\end{definition}

\begin{fact}
Given \SEP\ constraints, no polynomial time algorithm can beat the performance of a random permutation, which  achieves a $\tfrac13$-approximation, assuming \UGC. In fact, if $K$ is the total number of constraints, for any $\epsilon>0$, it is UGC-hard to distinguish between the two cases:\\
\textbf{YES}: $val(\pi^*)\ge (1-\epsilon)K$, i.e. the optimal permutation satisfies almost all constraints. \\
\textbf{NO}: $val(\pi^*) \le (\tfrac13 +\epsilon)K$, i.e. the optimal permutation does not satisfy more than 1/3 fraction of the constraints.
\end{fact}

Observe that from the $4!=24$ permutations on $a,b,c,d$ only $8$ of them obey the \SEP\ constraint, that's why random achieves $\tfrac13$. 

\begin{theorem}
Let $K$ be the total number of the quartet constraints in an instance of Desired Quartets Consistency. For any $\delta>0$, it is UGC-hard to distinguish between instances of the following two cases:\\
\textbf{YES}: $val(T^*)\ge (1-\delta)K$, i.e. the optimal tree satisfies almost  all the quartet constraints. \\
\textbf{NO}: $val(T^*)\le (\tfrac23 +\delta)K$, i.e. the optimal tree does not satisfy more than a $\tfrac23$ fraction of the quartet constraints.
\end{theorem}
\begin{proof}
We will make a reduction from the \SEP\ problem. Start from a \textbf{YES} instance and consider the optimum permutation $\pi^*$. Construct an unrooted caterpillar tree $T$ with leaves the labels of $\pi^*$ as they appear in the permutation. It is easy to see that if a \SEP\ constraint $ab|cd$ was obeyed by the permutation, then the corresponding quartet constraint $ab|cd$ was also obeyed in the caterpillar tree $T$. For that, we can assume  w.l.o.g. that the elements appear with relative order $abcd$ in $\pi^*$ and observe that the paths $a\to b$ and $c\to d$ in $T$ are disjoint, so the quartet is obeyed.

The harder case is the  \textbf{NO} instance. For that we will show how from a tree $T$ with high value, we can construct a permutation $\pi$ with high value. Specifically, we will show that if $val(T)> (\tfrac23 +2\epsilon)K$ then we can find $\pi$ with $val(\pi)> \tfrac12 (\tfrac23 +2\epsilon)K = (\tfrac13 +\epsilon)K$, a contradiction since we started from a  \textbf{NO} instance.

The tree $T$ is an unrooted tree on $n\ge 4$ leaves, whose internal nodes have degree exactly 3. We can make $T$ rooted by selecting an arbitrary internal node $r$ and making it the root of a binary tree whose internal nodes have exactly 2 children and one parent. The only exception is the root $r$ that has 3 children and no parent. Call this tree $T_r$. Let $A,B,C$ denote the leftmost, middle and rightmost child of $r$ respectively, which are themselves rooted binary trees. Assume w.l.o.g. that $A$ contains the largest number of leaves among $A,B,C$, so $|A|\ge2$, where $|A|$ denotes the number of leaves contained in the subtree rooted at $A$. 

From this rooted tree $T_r$, we generate a permutation $\pi$ by randomly swapping every left and right child on each internal node of $T_r$ and also randomly swapping $A,B,C$ at the root $r$; then we simply project the leaves onto a line to get $\pi$. We show that each quartet $q_1q_2|q_3q_4$ obeyed by $T$ will be obeyed in $\pi$ with probability $p\ge \tfrac12$. We have several cases depending on the labels $q_1,q_2,q_3,q_4$:
\begin{itemize}
    \item If $q_1,q_2\in A$ and $q_3\in B$ and $q_4\in C$: Notice that the status of the quartet is decided by the random choices at the root $r$ since after the final projection, labels from $A$ will be consecutive in $\pi$ and similarly for $B$ and $C$. Here, $\pi$ will actually obey the quartet with probability $\tfrac23$, as there are 3 equally likely outcomes $ABC$, $BCA$ and $CAB$ and the first two $ABC$ and $BCA$ obey the quartet, irrespectively of how labels from $A$, $B$, $C$ are ordered.
    
    \item If $q_1,q_2\in A$ and $q_3,q_4\in B$: This is the easiest case as every quartet of this form will be obeyed in $\pi$ with probability 1. This follows as labels from $A$ will be consecutive in $\pi$ and similarly for $B$.
    \item If $q_1,q_2,q_3\in A$ and $q_4\in B$: The status of this quartet only depends on how the elements $q_1,q_2,q_3$ are placed. Specifically, depending on the random choices at the root $r$, $q_4$ can appear either first (if $BA$ was chosen) or last (if $AB$ was chosen) among the 4 elements in $\pi$. If the former is true, then $q_3$ should appear second and we get $q_4q_3|\cdot\cdot$ otherwise $q_3$ should appear third and we get $\cdot\cdot|q_3q_4$. We need to compute the probability for each of these events. Notice that the lowest common ancestor both for $q_3,q_1$ and for $q_3,q_2$ is $A$. Hence, the status of the quartet is determined at $A$ and with probability $\tfrac12$, $q_3$ is correctly placed on the same side as $B$ (and $q_4$).
    \item If $q_1,q_2,q_3,q_4\in A$: This case essentially reduces to the analyses of the previous two cases. Just find the lowest common ancestor $A_1$ of all 4 labels $q_1,q_2,q_3,q_4$ in $T_r$. If two of the labels belong to one child and the remaining to the other child, then the quartet will be obeyed with probability $1$, irrespectively of the random choices at $A_1$ (similar to the second case above). Moreover, if one child contains three of the 4 elements, then the analysis is the same as the previous case yielding a probability of $\tfrac12$.
\end{itemize}
The other cases are symmetric for $B,C$. This proves that if a quartet is obeyed by the tree then with probability $\tfrac12$ will be obeyed in $\pi$ which means that $val(\pi)> \tfrac12 (\tfrac23 +2\epsilon)K = (\tfrac13 +\epsilon)K$ by linearity of expectation. This contradicts the fact that we were given a \textbf{NO} instance.
\end{proof}

\subsection{Hardness for Forbidden Quartets Consistency}
The proof proceeds in the same way as the previous paragraph, where we now account for the forbidden quartets and we make use of the complement problem to \SEP, which we call \NSEP:

\begin{definition}
For an ordering problem, a \textsc{4-Non-Separatedness} constraint $\{ab|cd\}$ specifies that either $a$ or $b$ should be between $c,d$ or that either $c$ or $d$ should be between $a,b$ in the output ordering (e.g., $adcb$, but not $abcd$). No constraints are placed on the relative ordering between $a,b$ or on the ordering between $c,d$.
\end{definition}

\begin{fact}
Given \NSEP\ constraints, no polynomial time algorithm can beat the performance of a random permutation, which  achieves a $\tfrac23$-approximation, assuming \UGC. In fact, if $K$ is the total number of constraints, for any $\epsilon>0$, it is UGC-hard to distinguish between the two cases:\\
\textbf{YES}: $val(\pi^*)\ge (1-\epsilon)K$, i.e. the optimal permutation satisfies almost all constraints. \\
\textbf{NO}: $val(\pi^*) \le (\tfrac23 +\epsilon)K$, i.e. the optimal permutation does not satisfy more than 2/3 fraction of the constraints.
\end{fact}

Observe that from the $4!=24$ permutations on $a,b,c,d$, $16$ of them obey the \NSEP\ constraint, that's why random achieves $\tfrac23$. 

\begin{theorem}
Let $K$ be the total number of the quartet constraints in an instance of Forbidden Quartets Consistency. For any $\delta>0$, it is UGC-hard to distinguish between instances of the following two cases:\\
\textbf{YES}: $val(T^*)\ge (1-\delta)K$, i.e. the optimal tree satisfies almost  all the quartet constraints. \\
\textbf{NO}: $val(T^*)\le (\tfrac89 +\delta)K$, i.e. the optimal tree does not satisfy more than a $\tfrac89$ fraction of the quartet constraints.
\end{theorem}
\begin{proof}
We will make a reduction from the \NSEP\ problem. Start from a \textbf{YES} instance and consider the optimum permutation $\pi^*$. Construct an unrooted caterpillar tree $T$ with leaves the labels of $\pi^*$ as they appear in the permutation. It is easy to see that if a \NSEP\ constraint $ab|cd$ was disobeyed (hence successfully avoided) by the permutation, then the corresponding quartet constraint $ab|cd$ was also disobeyed (i.e., avoided) in the caterpillar tree $T$. For that, we can assume  w.l.o.g. that the elements appear with relative order $acbd$ in $\pi^*$ and observe that the paths from $a\to c$ and from $b\to d$ in $T$ are disjoint, so the quartet is disobeyed as we wanted.

The harder case is the  \textbf{NO} instance. For that we will show how from a tree $T$ with high value, we can construct a permutation $\pi$ with high value. Specifically, we will show that if $val(T)> (\tfrac89 +\tfrac43\epsilon)K$ then we can find $\pi$ with $val(\pi)> \tfrac34 (\tfrac89 +\tfrac43\epsilon)K = (\tfrac23 +\epsilon)K$, a contradiction since we started from a  \textbf{NO} instance.

The tree $T$ is an unrooted tree on $n\ge 4$ leaves, whose internal nodes have degree exactly 3. We follow the same algorithm to generate the rooted $T_r$ and the final permutation $\pi$ as above. the notation for $A,B,C$ is the same as previously. We show that each quartet $q=q_1q_2|q_3q_4$ disobeyed by $T$ will be disobeyed in $\pi$ with probability $p\ge \tfrac34$. We have several cases depending on the labels $q_1,q_2,q_3,q_4$:
\begin{itemize}
    \item If $q_1,q_3\in A$ and $q_2\in B$ and $q_4\in C$: First notice that indeed quartet $q=q_1q_2|q_3q_4$ is disobeyed by the unrooted tree since it instead obeys $q_1q_3|q_2q_4$. We show that the probability that $\pi$ disobeys $q$ is $\tfrac56$. If the random choices at the root $r$ produce $ABC$ or $BCA$, then with probability $1$ the quartet $q$ is disobeyed after the projection. For example, if the realization is $ABC$ notice that either $q_3$ will be between $q_1,q_2$ or $q_1$ will be between $q_3,q_4$, thus disobeying the corresponding \NSEP\ constraint. Symmetrically, we handle the scenario where the realization was $BCA$. However, with probability $\tfrac13$ the realization at the root was $CAB$ and now the status of the quartet is determined by the random choice at the lowest common ancestor of $q_1,q_3$. With probability $\tfrac12$ label $q_1$ precedes $q_3$, thus giving the ordering $q_4q_1|q_3q_2$ disobeying $q$. In total $q$ is avoided with probability $\tfrac23+\tfrac16=\tfrac56$.
    \item If $q_1,q_3\in A$ and $q_2,q_4\in B$: This is the easiest case as every quartet of this form will be disobeyed in $\pi$ with probability 1. This follows as labels from $A$ will be consecutive in $\pi$ and similarly for $B$.

    \item If $q_1,q_2,q_3\in A$ and $q_4\in B$: The status of this quartet only depends on how the elements $q_1,q_2,q_3$ are placed. Specifically, depending on the random choices at the root $r$, $q_4$ can appear either first (if $BA$ was chosen) or last (if $AB$ was chosen) among the 4 elements in $\pi$. If the former is true, then $q_3$ should appear third or fourth  and we get $q_4\cdot|q_3\cdot$ or $q_4\cdot| \cdot q_3$, otherwise $q_3$ should appear first or second and we get $q_3\cdot|\cdot q_4$ or $\cdot q_3| \cdot q_4$. We need to compute the probability for each of these events. By the fact that the tree $T$ disobeys $q$, we can assume w.l.o.g. that label $q_1$ is the closest to $q_3$, otherwise we just rename $q_2$ as $q_1$ and vice versa. We get that the lowest common ancestor $A_{13}$ of $q_1,q_3$ in $A$ is strictly lower than the lowest common ancestor $A_{12}$ of $q_1,q_2$ in $A$ (in terminology of triplets consistency we have $q_1q_3|q_2$). W.l.o.g. assume that $BA$ was chosen at the root $r$, so $q_4$ will appear first. By the random choice in our algorithm, $A_{12}$ placed $q_2$ at the left child (hence second among the 4 elements) with probability $\tfrac12$ and the quartet $q$ is disobeyed. If instead our algorithm placed $q_2$ at the right child (and hence fourth in the ordering), there is still $\tfrac14$ probability of placing $q_3$ at the right child of $A_{13}$. This means that with probability $\tfrac12+\tfrac14=\tfrac34$, the projected $\pi$ disobeys $q$ as promised by the theorem.

    \item If $q_1,q_2,q_3,q_4\in A$: This case essentially reduces to the analyses of the previous two cases. Just find the lowest common ancestor $A_1$ of all 4 labels $q_1,q_2,q_3,q_4$ in $T_r$. If two of the labels belong to one child and the remaining to the other child, then the quartet will be disobeyed with probability $1$, irrespectively of the random choices at $A_1$ (similar to the second case above). Moreover, if one child contains three of the 4 elements, then the analysis is the same as the previous case yielding a probability of $\tfrac34$.
\end{itemize}
The other cases are symmetric for $B,C$. This proves that if a quartet is disobeyed by the tree then with probability $\tfrac34$ will be disobeyed in $\pi$ which means that $val(\pi)> \tfrac34 (\tfrac89 +\tfrac43\epsilon)K = (\tfrac23 +\epsilon)K$ by linearity of expectation. This contradicts the fact that we were given a \textbf{NO} instance.
\end{proof}

\end{document}